\documentclass[a4paper,UKenglish,cleveref, autoref, thm-restate]{lipics-v2021}
\usepackage{booktabs}
\usepackage{xspace}
\usepackage{comment}
\usepackage{thmtools}
\usepackage{booktabs}
\nolinenumbers
\hideLIPIcs  

\usepackage{microtype} 

\newcommand\polylog{{\rm polylog}}
\graphicspath{ {figures/} }

\newcommand {\R} {\mathbb{R}\xspace}
\newcommand {\Q} {\mathbb{Q}\xspace}

\newtheorem*{conjecture*}{Conjecture}

\newcommand{\mkmcal}[1]{\ensuremath{\mathcal{#1}}\xspace}

\newcommand{\F}{\mkmcal{F}}

\newcommand{\T}{\mkmcal{T}}

\newcommand{\Matousek}{Matou{\v s}ek\xspace}

\newcommand{\eps}{\ensuremath{\varepsilon}\xspace}

\newcommand{\Vor}{\ensuremath{\mathit{Vor}}\xspace}

\newcommand{\geod}{\ensuremath{\pi}\xspace}
\newcommand{\dist}{\ensuremath{d}\xspace}
\newcommand{\Gr}{\ensuremath{G^Q}\xspace}
\newcommand{\anchorDistancePoints}[2]{\hat{\dist}_Q(#1, #2)}

\def\polylog{\operatorname{polylog}}

\DeclareMathOperator*{\argmin}{arg\,min}

\newcommand{\updatecomplexityeps}{\ensuremath{O(\frac{1}{\varepsilon^2}\log n + \frac{1}{\varepsilon}\log n \log m + \frac{1}{\varepsilon}\log^2 m)}}
  \newcommand{\updatecomplexityk}{\ensuremath{O(k^2\log n + k\log n \log m + k\log^2 m)}}
\newcommand{\querycomplexityk}{\updatecomplexityk}
\newcommand{\querycomplexityeps}{\updatecomplexityeps}

\newcommand{\deletioncomplexityeps}{\ensuremath{O(\frac{1}{\varepsilon}\log n\log m + \frac{1}{\varepsilon}\log^2 m)}}
\newcommand{\deletioncomplexityk}{\ensuremath{O(k\log n\log m + k\log^2 m)}}

\newcommand{\myremark}[4]{\textcolor{blue}{\textsc{#1 #2: }}\textcolor{#4}{\textsf{#3}}}
\renewcommand{\myremark}[4]{}

\newcommand{\joost}[2][says]{\myremark{Joost}{#1}{#2}{teal}}

\newcommand{\frank}[2][says]{\myremark{Frank}{#1}{#2}{magenta}}

\title{Approximate Dynamic Nearest Neighbor Searching in a Polygonal Domain}

\author{Joost van der Laan}{Department of Information and Computing Sciences, Utrecht University, The Netherlands}{j.vanderlaan1@students.uu.nl}{}{}

\author{Frank Staals}{Department of Information and Computing Sciences, Utrecht University, The Netherlands}{f.staals@uu.nl}{https://orcid.org/0009-0004-8522-1351}{}

\author{Lorenzo Theunissen}{Delft University of Technology, The Netherlands}{L.S.J.Theunissen@tudelft.nl}{https://orcid.org/0009-0008-3405-541X}{}

\authorrunning{J. van der Laan, F. Staals, and L. Theunissen} 

\ccsdesc[100]{Theory of computation~Computational Geometry} 

\keywords{dynamic data structure, nearest neighbor search, polygonal domain
}  

\Copyright{Joost van der Laan, Frank Staals, and Lorenzo Theunissen} 

\ArticleNo{100}

\begin{document}

\maketitle

\begin{abstract}
  We present efficient data structures for approximate nearest
  neighbor searching and approximate 2-point shortest path queries in
  a two-dimensional polygonal domain $P$ with $n$ vertices. Our goal
  is to store a dynamic set of $m$ point sites $S$ in $P$ so that we
  can efficiently find a site $s \in S$ closest to an arbitrary query
  point $q$. We will allow both insertions and deletions in the set of
  sites $S$. However, as even just computing the distance between an
  arbitrary pair of points $q,s \in P$ requires a substantial amount
  of space, we allow for approximating the distances. Given a
  parameter $\varepsilon > 0$, we build an
  $O(\frac{n}{\varepsilon}\log n)$ space data structure that can
  compute a $1+\varepsilon$-approximation of the distance between $q$
  and $s$ in $O(\frac{1}{\varepsilon^2}\log n)$ time. Building on
  this, we then obtain an
  $O(\frac{n+m}{\varepsilon}\log n + \frac{m}{\varepsilon}\log m)$
  space data structure that allows us to report a site $s \in S$ so
  that the distance between query point $q$ and $s$ is at most
  $(1+\varepsilon)$-times the distance between $q$ and its true
  nearest neighbor in
  $\querycomplexityeps$ time. Our data structure supports updates in
  $\updatecomplexityeps$ amortized time.
\end{abstract}

\section{Introduction}

Nearest neighbor searching is a fundamental problem in which the goal
is to store a set $S$ of $m$ point sites so that given a query point
$q$ one can quickly report a site in $S$ that is closest to $q$. For
example, the set $S$ may be a set of (locations of) shops, and we may
wish to efficiently find the shop closest to our current location. In
addition, nearest neighbor searching shows up as an important
subroutine in several other fundamental problems. For example, in
computing closest pairs~\cite{chan20dynam_gener_closes_pair} and
matchings~\cite{kaplan20dynam_planar_voron_diagr_gener}. In these
applications, there are often two (additional) difficulties: (i) the
point sites are typically embedded in some constrained domain, and
(ii) the set of sites may change over time, and hence it may be
necessary to insert and delete sites. In the example problem above
other buildings, roads, and lakes etc. act as obstacles, and thus we
actually have to measure the distance $\dist(q,s)$ between two points
$q$ and $s$ in terms of the length of a shortest obstacle avoiding
path $\geod(q,s)$. Furthermore, shops may open and close (or may be
too crowded or out of products etc.), and hence we want to consider
only the currently open or available shops when querying.

In this paper, we consider dynamic nearest neighbor searching among a
set $S$ of $m$ sites in a two-dimensional polygonal domain $P$ with
$n$ vertices. We will allow both insertions and deletions in the set
of sites $S$. Agarwal, Arge, and,
Staals~\cite{agarwal18improv_dynam_geodes_neares_neigh} presented a
near linear space data structure for this problem that allows
efficient (polylogarithmic) queries and updates, provided that $P$
contains no holes; i.e. it is a simple polygon. However, extending
these results to arbitrary polygonal domains seems to be far out of
reach. It seems we need to at least be able to efficiently compute the
(geodesic) distance $\dist(q,s)$ between two arbitrary query points
$q$ and $s$ in $P$. Currently, the best data structures that support
such queries in polylogarithmic time use $O(n^9)$
space~\cite{berg24towar_space_effic_two_point} (and it seems unlikely
we can determine the site closest to $q$ without actually evaluating
geodesic distances). Hence, we settle on answering queries
approximately. Let $\eps > 0$ be some parameter, and let $s^*$ be the
site in $S$ closest to a query point $q$, then our goal is to report a
site $s \in S$ whose distance to $q$ is at most
$(1+\eps)\dist(q,s^*)$. We will say that $s$ is \emph{$\eps$-close} to
$q$. Our main result is then:

\begin{restatable}{theorem}{mainresult}
  \label{thm:main}
  Let $P$ be a polygonal domain with $n$ vertices, let $S$ be a
  dynamic set of $m$ point sites in $P$, and let $\eps > 0$ be a
  parameter. There is an
  $O(\frac{n+m}{\eps}\log n + \frac{m}{\eps}\log m)$ space data structure that can
  \begin{itemize}
  \item compute a site in $S$ that is $\eps$-close to a query $q \in P$ in
    $\querycomplexityeps$ time,
  \item support insertions of sites into $S$ in
    $\updatecomplexityeps$ amortized time, and
  \item support deletions of sites from $S$ in
    $\deletioncomplexityeps$ amortized time.
  \end{itemize}

  \noindent
  Constructing the initially empty data structure takes
  $O(\frac{n}{\eps^2}\log^2 n)$ time.
\end{restatable}

The crucial ingredient in our data structure is a data structure for
$(1+\eps)$-approximate 2-point shortest path
queries. Thorup~\cite{Thorup2007} presents an
$O(\frac{n}{\eps}\log n)$ space data structure that, given two
arbitrary query points $s,t$ in the polygonal domain $P$, can compute
an estimate $\hat{\dist}(s,t)$ of their distance so that
$\dist(s,t) \leq \hat{\dist}(s,t) \leq (1+\eps)\dist(s,t)$ in
$O(\frac{1}{\eps^3} + \frac{\log n}{\eps\log\log n})$ time. Thorup's
solution however assumes that distances are encoded as floating point
numbers whose bits we can manipulate. This is inconsistent with the
Real-RAM model of computation usually assumed in computational
geometry. We stick to the Real-RAM model, and show that we can still
compute approximate distances efficiently. We can actually improve the
dependency on $\eps$, and obtain:

\begin{restatable}{theorem}{improvedoracle}
  \label{thm:oracle}
  Let $P$ be a polygonal domain with $n$ vertices, and let $\eps > 0$
  be a parameter. In $O(\frac{n}{\eps^2}\log^2 n\log\frac{n}{\eps})$
  time we can construct an $O(\frac{n}{\eps}\log n)$ space data
  structure that, given query points $s,t \in P$ returns an
  $(1+\eps)$-approximation of $\dist(s,t)$ in
  $O(\frac{1}{\eps^2}\log n)$ time.
\end{restatable}

Our overall approach is actually the same as in Thorup's approach: we
recursively construct a separator, and estimate distances by going via
a discrete set of anchor points on this separator. However, our
approach provides stronger guarantees on these anchor points. We
believe this may be useful for other problems as well. Indeed, it
allows us to efficiently
answer (approximate) nearest neighbor queries as well. This also still
requires several additional ideas, as simply evaluating the distance
between a query point and every site would only give us query time
linear in $m$.

\subparagraph{Related Work.} Dynamic nearest neighbor searching in the
Euclidean plane has a long history, with early results by Agarwal and
\Matousek~\cite{agarwal1995dynamic}. Chan~\cite{chan10d_d} was the
first to achieve expected polylogarithmic update and query
times. Following a sequence of improvements, the best structure at
this time uses $O(m)$ space, allows for $O(\log^2 m)$ time queries,
and for updates in deterministic, amortized $O(\log^4 m)$
time~\cite{kaplan20dynam_planar_voron_diagr_gener,chan20dynam_geomet_data_struc_shall_cuttin}. For
more general (constant complexity) distance functions, one can
similarly achieve expected polylogarithmic query and update
times~\cite{kaplan20dynam_planar_voron_diagr_gener,
  liu22nearl_optim_planar_neares_neigh}. When the domain is a simple
polygon $P$ with $n$ vertices, Oh and
Ahn~\cite{oh20voron_diagr_moder_sized_point} presented a solution
using $O(n+m)$ space achieving $O(\sqrt{m}\log (n+m))$ query time, and
$O(\sqrt{m}\log m\log^2 n)$ update time. Agarwal, Arge, and,
Staals~\cite{agarwal18improv_dynam_geodes_neares_neigh} then showed
how to achieve $O(\log^2 n\log^2 m)$ query time and expected
$O(\polylog (nm))$ update time using $O(n + m\log^3 m\log n)$
space. When the domain may have holes and the set of sites is static,
there is an $O(n+m)$ space data structure that allows querying the
(exact) nearest neighbor in $O(\log (n+m))$
time~\cite{HershbergerSuri1999}. Constructing the data structure uses
$O(n\log n)$
time. Wang~\cite{wang23new_algor_euclid_short_paths_plane} recently
showed that if $P$ is triangulated, the shortest path map of a single
site can be computed in $O(n + h\log h)$ time, where $h$ is the number
of holes in $P$. Conceivably, his algorithm could be extended to
handle multiple sites (as is the case in the algorithm by Hershberger
and Suri). We are not aware of any results that allow inserting or
deleting sites. Even the restricted case where we allow only
insertions is challenging: we cannot easily apply a static-to-dynamic
transformation~\cite{bentley1980decomposable} since rebuilding the
static structure may require $\Omega(n)$ time.


\section{Global Approach}
\label{sec:Global_Approach}

We first present an overview of our data structure. See also
Figure~\ref{fig:intro-overview}. We start by
briefly reviewing the $\eps$-approximate 2-point shortest path query
data structure by Thorup~\cite{Thorup2007}.

\begin{figure}
  \centering
  \includegraphics[page=2]{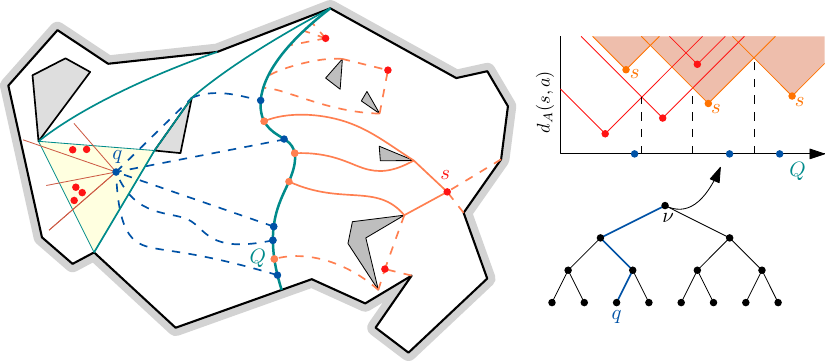}
  \caption{An overview of our approach. We use shortest paths to
    recursively subdivide $P$ into subpolygons. If $s^*$ lies in the
    same triangle as $q$, we can directly compute an $\eps$-close
    site. Otherwise, the shortest path must cross a shortest path $Q$,
    and we find a site $s$ that is sufficiently close to $q$ whose
    path goes via an anchor point on $Q$.
  }
  \label{fig:intro-overview}
\end{figure}

\subparagraph{Computing approximate shortest paths.} Thorup's approach
recursively partitions the domain $P$ using \emph{separators} such
that: each separator consists of at most three shortest
paths\footnote{Thorup\cite{Thorup2007} writes that there are at most
  six such shortest paths. As observed in \cite{berg24spanners}, there
  three shortest paths actually suffice.}, and the resulting subpolygons are
of roughly the same size. This results in a balanced binary tree \T of
height $O(\log n)$ in which each node $\nu$ corresponds to some
subpolygon $P_\nu$ that is further subdivided using a separator
$Q_\nu$. 

Consider the separator $Q$ of the root of this tree \T, and let
$k=\lceil 1/\eps \rceil$. The main idea is that each polygon vertex
$v$ generates a set $A_Q(v)$ of $O(k)$ so called \emph{anchor
  points}\footnote{Thorup calls these points connections, but we use
  anchor points to avoid confusion with the paths connecting $v$ to
  these points.} on $Q$, so that for any point $q \in Q$ the distance
via one of these anchor points $A_Q(v)$ to $v$ is at most $(1+\eps)$
times the true distance $\dist(q,v)$. So, if a shortest path
$\geod(u,v)$ between two vertices $u$ and $v$ intersects $Q$
then there is a path from $u$ to $v$ via two anchor points
$a \in A_Q(u)$ and $a' \in A_Q(v)$ (and the subcurve $Q[a,a']$ of $Q$)
of length $(1+\eps)\dist(u,v)$. Given the anchor
points (and the distances to their defining vertices) we can find the
length of such a path in $O(k)$ time. If the path $\geod(u,v)$ does
not intersect $Q$, it is contained in one of the subpolygons of the
children of the root, and hence we can compute a
$(1+\eps)$-approximation of the distance $\dist(u,v)$ in the
appropriate subpolygon.


As each vertex defines $O(k)$ anchor points on $Q$, and appears in
$O(\log n)$ levels of the recursion. The total space used is
$O(nk\log n)$.

To compute the anchor points efficiently, Thorup's actual algorithm
uses a graph $G$ in which the distance $\dist_G(u,v)$ between $u$ and
$v$ in the graph is already an $(1+\eps_1)$-approximation of the true
distance $\dist(u,v)$. Using the above approach with parameter
$\eps_0$ then actually gives a path in $P$ with total length at most
$(1 + \eps_0)(1 + \eps_1)\dist(u,v)$. Setting
$\eps_0=\eps_1 = \eps/3$ then gives a $\eps$-approximation of the
distances as desired.

However, using the graph $G$ results in slightly different (generally
weaker) properties compared to the ideal case sketched above. In
particular, his approach now guarantees that the distances from vertex $v$ are within a
factor $1+\eps$ of the true distance only for a discrete set
of points on $Q$ (the vertices of $Q$ and the intersections between
edges of $G$ and $Q$). For querying the distances
between pairs of vertices this is not an issue. Thorup's data
structure can handle such queries in only $O(k)$ time. To support
queries between arbitrary points $s,t \in P$, his algorithm connects
$s$ and $t$ to $O(k)$ vertices in $G$ each, and then simply queries
the distance between all pairs of these vertices. This leads to the
query time $O(k^3 + k\log n/\log\log n)$.

\subparagraph{Reducing the query time.} Ideally, we would treat $s$
and $t$ identical to polygon vertices, so that we get $O(k)$ anchor
points on a separator $Q$. However, as such a candidate anchor point
$p \in Q$ may lie anywhere on $Q$ (in particular, it is not
necessarily a vertex of $G$ yet), we may not have the required
$\eps_1$-approximation of the distance to such a point yet.

We will show that we can use a slightly different graph $\Gr$ that
includes additional ``Steiner'' points on $Q$.  This leads to a
different set of anchor points compared to Thorup's approach, but
\emph{does} guarantee that for any vertex $v$, and \emph{any} point
$q$ on $Q$, the distance to $v$ is at most
$(1+\eps)\dist(q,v)$. Furthermore, we show that we can efficiently compute an initial set of $O(k^2)$ candidate anchor
points for an arbitrary point $s$, which we then reduce to
$O(k)$. This results in $O(k^2\log n)$ query time for an arbitrary
pair of query points $s,t \in P$.

\subparagraph{Answering nearest neighbor queries.} As in the
approximate 2-point shortest path data structure above, we construct a
tree \T of separators. We associate every site $s \in S$ with a
root-to-leaf path in this tree. Namely, the nodes $\nu$ for which $s$
is contained in the subpolygon $P_\nu$. For each node $\nu$ on
this path, we will additionally maintain a set of anchor points $A_Q(s)$
on the separator $Q = Q_\nu$ corresponding to this node. In
particular, we consider $Q$ as a one dimensional space, and maintain
an additively weighted Voronoi diagram $\Vor(A_Q)$ on $Q$. The set of
sites $A_Q = \bigcup_{s \in S \cap P_\nu} A_Q(s)$ is the set of all
anchor points of all sites that appear in $P_\nu$. The weight of an
anchor point $a \in A_Q(s)$ is the distance estimate from $a$ to $s$.

If a shortest path between a query point $q$ and its closest site
$s^* \in S$ intersects $Q$, then there is again a sufficiently short
path via one of the anchor points of $s^*$ and $q$ on $Q$. We can then
find a sufficiently close site $s$ by querying the Voronoi diagram
$\Vor(A_Q)$ with all anchor points of $q$. We can answer each of these
$O(k)$ queries in $O(\log (km))$ time. This leads to an overall query time of
$\querycomplexityk$, as we need to perform such queries
in $O(\log n)$ nodes in \T.

The graph of the function that expresses the distance from a site
$s \in S$ to a point $p$ on $Q$ via a given anchor point
$a \in A_Q(s)$ has a nice ``V-shaped'' structure (see
Figure~\ref{fig:intro-overview}). Hence, we show that we can maintain
$\Vor(A_Q)$ by maintaining the lower envelope of these distance
functions. This allows us to insert and remove sites in
$O(k\log^2 (km))$ time. This leads to an overall update time of
$\updatecomplexityk$. Our overall
structure uses $O(nk\log n + mk\log n + mk\log m)$ space.


\subparagraph{Organization.} In Section~\ref{sec:cones}, we define our
graph $G^Q$, and prove that it allows us to accurately estimate
distances between any pair of points $s,t$ whose shortest path
intersects $Q$. We then use this to construct our approximate 2-point
shortest path data structure in
Section~\ref{sec:An_improved_data_structure_for_distance_queries}. Finally,
in Section~\ref{sec:Dynamic_nearest_neighbor_searching} we describe
our lower envelope data structure and the remaining tools needed for an efficient dynamic data
structure for approximate nearest neighbor searching.

\section{Estimating distances using a graph}
\label{sec:cones}

In this section we discuss how to approximate distances within a
polygon $P$. Our results are based on Clarkson's cone
graph~\cite{Clarkson1987} on which we can efficiently compute
approximate distances between vertices. We review this result in
Section~\ref{sub:Clarkson's_Cone_graph}. In
Section~\ref{sub:Continuous_eps-approximation_graphs} we then extend
this approach so that we can represent the distance from any vertex of
$P$ to any point on a given shortest path $Q$ in $P$. This involves
computing the anchor points for each vertex of $P$. Finally, in
Section~\ref{sub:The_augmented_cone_graph}, we further augment the
graph to also represent distances from a set $S$ of arbitrary point
sites to (and via) the shortest path $Q$, and show how to
compute the anchor points for such an arbitrary point $s \in S$ efficiently.





\begin{figure}
    \centering
    \includegraphics[page=2]{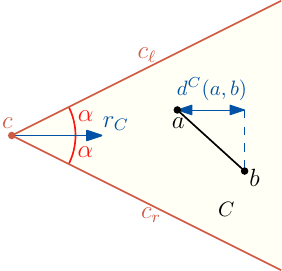}
    \quad
    \includegraphics[page=4]{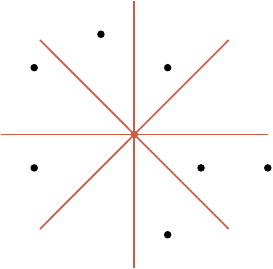}
    \caption{(a) A cone $C$ with apex $c$, cone direction $r_C$, and
      angle $\alpha$, and an illustration of the cone distance
      $d^C(a,b)$ between two points $a,b \in C$.  (b) A family of cones with
      angle $2\alpha = 45$ degrees.  }
    \label{fig:cone_and_cone_family}
\end{figure}

\begin{figure}
  \centering
  \includegraphics[page=5]{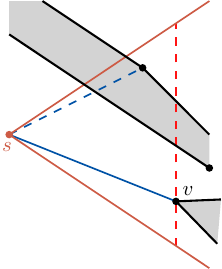}
  \quad
  \includegraphics[page=6]{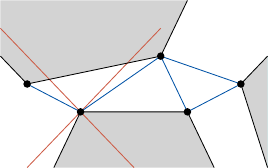}
  \caption{(a) A point $s$, which has an outgoing cone edge to a
    vertex $v$ w.r.t the pictured cone $C$. Hence, there are no
    visible vertices in $C$ to the left of the red dashed line. (c)
    An example cone graph $G$ (for a very large $\eps$). 
  }
  \label{fig:cones}
\end{figure}

\subsection{Clarkson's cone graph}
\label{sub:Clarkson's_Cone_graph}

Clarkson~\cite{Clarkson1987}'s graph is based on cones. A \emph{cone}
$C$ is characterized by an apex point $c$, a cone direction $r_C$
which is a unit length vector and an angle $\alpha \leq \pi/2$.  The
boundary of the cone consists of two half lines $c_\ell, c_r$ who both
start at $c$ and the difference in angle compared to $r_C$ is
$\alpha/2$ and $-\alpha/2$ respectively. We use $C_c$ to denote a cone
$C$ to have its apex at point $c$. A \emph{cone family}
$\mathcal{F}$ is a set of cones that all have their apex at the
origin, such that any point in $\mathbb{R}^2$ is contained in at least
one cone $C\in \mathcal{F}$. See Figure~\ref{fig:cone_and_cone_family}(b)
for an example.

Let $\eps > 0$ be a parameter. Clarkson then defines a cone family
$\mathcal{F}_\eps$ consisting of $O(k)=O(1/\eps)$ cones, each with
angle at most $\eps/8$, and a \emph{cone graph} $G$ based on
$\mathcal{F}_\eps$.

Let $V$ be the set of vertices of $P$. A vertex $b \in V$ is a
\textit{minimal cone neighbor} of a point $a \in P$ if and only if
there exists a cone $C\in\mathcal{F}_\eps$, so that $b \in S$ is the
vertex in the translated cone $C_a$ with minimum
\emph{cone-distance} $d^C(a,b) = |r_C\cdot (b - a)|$ to $a$ among all
vertices that are visible from $a$ (i.e. for which the line segment
$\overline{ab}$ is contained in $P$). 
We will also
say that $a$ has an \textit{outgoing cone edge} to $b$ w.r.t the cone
$C$. Let $N(a)$ denote the set of these minimal cone neighbors of $a$
over all cones in $\mathcal{F}_\eps$. See Figure~\ref{fig:cones}(a)
for an example.

  The cone graph $G = (V,E)$ then has an edge $(u,v) \in E$ between
  two (polygon) vertices $u,v \in V$ if and only if $u$ had $v$ as its
  \emph{minimal cone neighbor} or vice versa, i.e. $v \in N(u)$ or
  $u \in N(v)$. Each vertex $u$ has at most $O(k)$ outgoing cone
  neighbors, and thus the total number of edges in $G$ is at most
  $O(nk)$. Let $\dist_G(u,v)$ denote the length of a shortest path
  $\geod_G(u,v)$ between vertices $u,v$ in $G$.

  \begin{lemma}[Clarkson~\cite{Clarkson1987}]
  \label{lemm:Clarkson}
  For every pair of vertices $u, v$ in the cone graph $G$, there is a
  path $\geod_G(u,v)$ of length $d_{G}(u,v)$ so that
  \[
    \dist(u,v) \leq \dist_G(u,v) \leq (1 + \eps)\dist(u,v).
  \]
  In $O(nk\log n)$ time we can build a data structure of size $O(nk)$
  so that for any point $s \in P$ one can compute the set of minimal
  cone neighbors $N(s)$ in $O(k\log n)$ time. Hence, the graph $G$ can
  be constructed in $O(nk\log n)$ time.
\end{lemma}

\subparagraph{Extending the cone graph.} Let $s \in P$ be a point, we
define the \emph{extended cone graph}
$G[s] = (V \cup \{s\}, E \cup \{(s,v) \mid v \in N(s)\})$ by adding
$s$ as a vertex, and connecting it to all of its outgoing cone
neighbors. Note that this graph is different from Clarkson's cone
graph on $V \cup \{s\}$. See Figure~\ref{fig:extended-cone-graph} for
an illustration. Using essentially same argument as in
Lemma~\ref{lemm:Clarkson}, we can still prove that $G[s]$ provides an
$\eps$-approximation of the distances to and from $s$. This uses the
following geometric observations (which we will also need later).

\begin{figure}
  \centering
  \includegraphics[page=7]{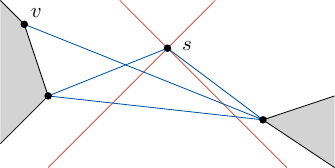}
  \includegraphics[page=8]{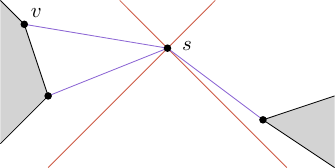}
    \caption{(a) The graph $G[s]$, where only outgoing cone edges from
      $s$ have been added. (b) The cone graph on $V \cup \{s\}$. Observe
      that the edge set of neither graph is a subset of the other.
    }
    \label{fig:extended-cone-graph}
\end{figure}

\begin{figure}[tb]
  \centering
  \includegraphics{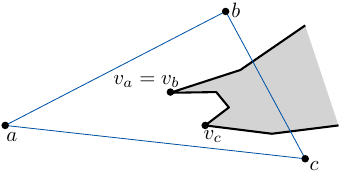}
  \caption{The configuration from
    Observation~\ref{obs:triangle_third_intersection}.}
  \label{fig:visible_vertex}
\end{figure}

\begin{observation}
  \label{obs:triangle_third_intersection}
  Let $a, b$ and $c$ be points in $P$. If $b$ and $c$ are both visible
  to $a$ but not to each other, then for each $x\in \{a,b,c\}$ there
  exists a vertex $v_x$ of $P$ in the triangle $\Delta abc$ with corners
  $a, b$, and $c$ such that
  (i) $v_x$ is visible to $x$, and (ii) $v_x$ does not lie on $\overline{bc}$.
  See Figure~\ref{fig:visible_vertex}.
\end{observation}

\begin{lemma}
  \label{lem:path_via_cone_neighbor}
  Let $s$ be a point in $P$, let $C \in \mathcal{F}_\eps$ be a cone
  from the cone family, and let $u$ and $q$ be two points in the cone
  $C_s$ that (i) are visible from $s$, and (ii) so that
  $d^C(s,u) \leq d^C(s,q)$, and (iii) there exists no vertex $v$ of the
  polygon which is both visible in $C$ and satisfies $d^C(s, v) < d^C(s,u)$. Then (1) the shortest obstacle avoiding path $\gamma^*$
  from $u$ to $q$ has length at most $\dist(s,q)$, and (2) for any obstacle avoiding path
  $\gamma$ from $u$ to $q$ of length at most $(1+\eps)\dist(u,q)$ the
  concatenation of the paths $\overline{su}$ and $\gamma$ is an
  obstacle avoiding path of length at most $(1+\eps)\dist(s,q)$.
\end{lemma}

\begin{proof}
We assume that $(i) - (iii)$ all hold and will demonstrate that we find
that $(1), (2)$ must hold. We rotate the plane such that the cone
direction of $C_s$ is pointing right along the $x$-axis. First we
consider a line through $u$ perpendicular to the cone direction of
$C_s$ (i.e. the line is vertical). At the points where this line
intersects the left and right cone boundaries and $\overline{sq}$ we
define the points $\ell, r$ and $m$, respectively. See
Figure~\ref{fig:Clarkson_image} for an illustration. Observe that $m$ must exist because $q$ lies on or to the right of our vertical line per assumption $(ii)$.

\begin{figure}
      \centering
      \includegraphics[page=9]{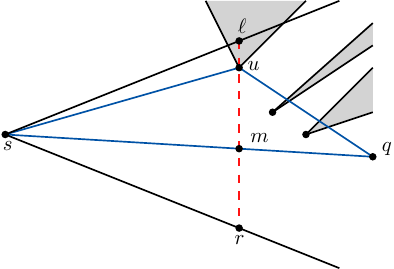}
      \caption{The points $s, u, q$ and $m$ as described in the proof
        of Lemma \ref{lem:path_via_cone_neighbor}.}
      \label{fig:Clarkson_image}
    \end{figure}

Per assumption $(i)$ the linesegments $\overline{su}$ and
$\overline{sq}$ are not intersected by the polygon boundary, and per assumption
$(iii)$ there are no vertices in $\Delta sum$ which are both: strictly
left to the linesegment $\overline{um}$ and visible to $s$. Therefore,
per Observation~\ref{obs:triangle_third_intersection} the segment
$\overline{um}$ is also not intersected by the boundary of the polygon.

As the path consisting of the line segments $\overline{um}$ and $\overline{mq}$ is an obstacle avoiding path from $u$ to $q$ we find that the length of $\gamma$ is at most
\begin{align*}
  (1 + \varepsilon) \dist(u, q) \leq (1 + \eps) (\dist(u, m) + \dist(m, q)).
\end{align*}
Next, we will relate the length of $\dist(u, m)$ to the length of
  $\dist(s, q)$, as all of these points are visible to each other
  these distances are simply the Euclidean distance.

  Per our construction we know that the angle of the cone $C_s$ is $\alpha = \eps/8$ and because $\eps \leq 1$ we also have that $\alpha \leq \pi/4$.
  Furthermore, we can derive
  \begin{align*}
    3\pi/4 \geq \angle sum \geq \angle s\ell r \geq \pi/4,
  \end{align*}
  and therefore $\sin(\angle sum) \geq 1/2$.
  Finally, using the law of sines on the triangle $\Delta sum$ we find that
  \begin{align*}
    \dist(u,m) &= \frac{\dist(s,m)\sin(\angle usm)}{\sin(\angle sum)}\\
               &\leq 2\dist(s,m)\sin(\alpha)\\
               &\leq \dist(s,m)\eps/4.
  \end{align*}
  Additionally, this proves $(1)$ as the length of $\gamma^*$ is at most $\dist(u,m) + \dist(m, q)$, and therefore shorter than $\dist(s, q)$.

  Combining the above inequality with our earlier derived equation for the length of $\gamma$ we find that the length of the concatenated path is at most
  \begin{align*}
    \dist(s, u) + (1 + \varepsilon) \dist(u, q) &\leq \dist(s,m) + \dist(m, u) + (1 + \eps)(\dist(u,m) + \dist(m,q))\\
                                 &= \dist(s,m) + (2 + \eps)\dist(m, u) + (1 + \eps)\dist(m,q)\\
                                 &\leq \dist(s,m) + (\eps/2 + \eps^2/4)\dist(s,m) + (1 + \eps)\dist(m,q)\\
                                 &\leq (1 + \eps)(\dist(s,m) + \dist(m,q))\\
                                 &= (1 + \eps)\dist(s,q)
  \end{align*}
  as desired.
\end{proof}

\begin{lemma}
  \label{lem:path_via_cone_neighbor_exists}
  Let $s$ be a point in $P$. For any vertex $w$ in $G[s]$ there is a
  path from $s$ to $w$ in $G[s]$ of length
  $\displaystyle
    \dist(s,w) \leq \dist_{G[s]}(s,w) \leq (1+\eps)\dist(s,w).
  $
\end{lemma}

\begin{proof}
  Consider the shortest obstacle avoiding path $v_1, v_2, \dots v_t$ from $s$ to $w$, where $w = v_t$. Let $C_s$ a cone of the cone family centered at $s$ in which $v_1$ is visible. Let $u$ be the vertex with minimum cone distance to $s$ in $C_s$ and therefore $u\in N(s)$. Per Lemma~\ref{lem:path_via_cone_neighbor} we may conclude that
  \begin{align*}
    \dist_{G[s]}(s,w) &\leq \dist(s, u) + \dist_G(u, w)\\
    &\leq \dist(s, u) + \dist_G(u, v_1) + \dist_G(v_1, w)\\
    &\leq \dist(s, u) + (1 + \varepsilon)\dist(u,v_1) + (1 + \varepsilon)\dist(v_1, w)\tag{Lemma~\ref{lemm:Clarkson}}\\
    &\leq (1 + \varepsilon)\dist(s, v_1) + (1 + \varepsilon)\dist(v_1, w)\tag{Lemma~\ref{lem:path_via_cone_neighbor}}\\
    &= (1 + \varepsilon)\dist(s, w). \qedhere
  \end{align*}
\end{proof}

\subsection{Continuous $\eps$-approximation graphs}
\label{sub:Continuous_eps-approximation_graphs}

Let $Q$ be a shortest path between two vertices in $P$. Our goal is to
define and compute, for each vertex $v$ of $P$, a small set of
weighted anchor points $A_Q(v)$ on $Q$, so that for any point $q$ on
$Q$ there is an obstacle avoiding path between $q$ and $v$ of length
at most $(1+\eps)\dist(q,v)$ via one of these anchor points. More
specifically, let $\hat{\dist}_a$ be the weight of an anchor point
$a$. We then want that there is an obstacle avoiding path
$\hat{\geod}_a$ from $v$ to $a$ of length $\hat{\dist}_a$, and that
$\displaystyle \dist_G(v,A_Q(v),q) := \min_{a \in A_Q(v)} \dist_G(a,q) +
\hat{\dist}_a$ is at most $(1+\eps)\dist(q,v)$.


Thorup~\cite{Thorup2007} argues that such a set of $O(k)$ anchor
points exists. Although his argument (Lemma 5 in his paper) is
constructive, it is unclear how to efficiently construct such a set of
anchor points. So, he constructs a set of anchor points with respect
to Clarkson's cone graph instead. In particular, let $G=(V,E)$ be a
cone graph constructed with parameter $\eps_1 = \eps/3$. Thorup
defines the graph $G'$ by subdividing every edge in $G$ that intersects $Q$ (at the intersection point), and connecting
consecutive vertices along $Q$, and proves

\begin{lemma}[Thorup~\cite{Thorup2007}]
  \label{lem:thorup_anchor_points}
  Let $H=(V,E)$ be an undirected weighted graph, let $Q$ be a shortest
  path in $H$, and let $k' \in \mathbb{N}$ be a parameter. For each
  vertex $u \in H$, we can compute a set $A_Q(u) \subseteq V$ of
  $O(k')$ discrete anchor points on $Q$, so that for any vertex
  $v \in Q$ we have
  \begin{align*}
    \dist_H(u,v) \leq \dist_H(u,A_Q(u),v) \leq (1 + 1/k')\dist_H(u,v).
  \end{align*}
  Computing all sets of anchor points takes a total of
  $O\left(k'\log |Q|\left(|V|\log|V| + |E|\right)\right)$ time.
\end{lemma}

\begin{proof}
  This is a rephrasing of the results in Lemma 10 of
  Thorup~\cite{Thorup2007}, which in turn is based on the result of
  Lemma 3.18 in Thorup~\cite{Thorup2004}. In particular, Lemma 10 of
  Thorup~\cite{Thorup2007} states that (i) the anchor points can be
  computed by invoking a single source shortest path algorithm on
  induced subgraphs of $H$, and (ii) a vertex of $H$ appears in at
  most $O(k'\log |Q|)$ of these calls.

  Lemma 11 of Thorup~\cite{Thorup2007} then plugs in an $O(|V|+|E|)$
  time algorithm to compute single source shortest
  paths~\cite{thorup00float_integ_singl_sourc_short_paths}, provided
  that all edge weights are positive floating point values. However,
  this assumption is incompatible with the usual Real RAM
  assumption. So, we use a regular, comparison based,
  $O(|V|\log|V|+|E|)$ time implementation of Dijkstra's algorithm
  instead.

  Let $H_1,.,.H_t$ be the induced subgraphs graphs of $H$ on which we
  run the single source shortest path computations, and let $V_i$ and
  $E_i$ be the number of vertices and edges in $H_i$,
  respectively. The total running time is then thus
  $\sum_{i=1}^t O(V_i\log V_i + E_i)$. As each vertex of $H$ appears
  in $O(k'\log|Q|)$ subgraphs we have
  $\sum_{i=1}^t O(V_i\log V_i) = O(\log |V| \sum_{i=1}^t V_i) =
  O(k'\log |Q| |V|\log |V|)$. According to the analysis in Lemmas 3.18
  and 3.12 of Thorup~\cite{Thorup2004}, we can similarly bound
  $\sum_{i=1}^t E_i$ by $O(|E|k'\log |Q|)$. The result follows.
\end{proof}

By choosing $1/k'= \eps/3$ and applying Lemma~\ref{lem:thorup_anchor_points} on the
graph $G'$ this then allows him to efficiently estimate the distances
between the \emph{vertices} of $G'$. In particular, if the shortest
path in $G'$ between two vertices $u,v \in G'$ intersects $Q$ then
one can find an $(1+\eps)$-approximation of their distance $\dist(u,v)$
in $P$ using one of the constructed anchor points in $O(k)$ time.

However, we want something stronger: we want an $(1+\eps)$-approximation
for the distance between any vertex in $G$ and \emph{any point} on
$Q$; i.e.~not just the vertices of $Q$ and the intersection points of
$Q$ with edges of $G$. Thorup's construction does not immediately
provide that, see for example Figure~\ref{fig:counterexample_throrup}
for an illustration.


\begin{figure}[tb]
  \centering
  \includegraphics[page=3]{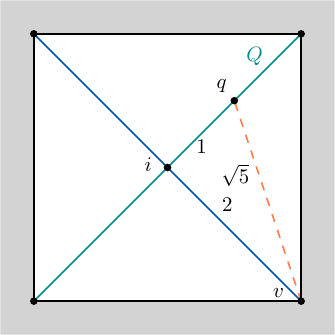}
\caption{An example that shows that there may be points $q$ on $Q$
    (shown in green) for which the graph $G'$ (black, blue and green
    edges, with $i$ subdividing an edge from the visibility graph $G$.) considered by Thorup does not yield a
    $\eps$-approximation. For any $\eps \leq 1/10$, Thorup's anchor
    points give us distance estimate of $3$, while the true distance
    is $\sqrt{5}$. Hence, this is only a
    $3/\sqrt{5} \geq 1.34$-approximation, not an
    $(1+\eps)$-approximation. The issue stems from the fact that $G'$ only consists of (parts of) edges in the visibility graph, which is independent of $\varepsilon$.}
  \label{fig:counterexample_throrup}
\end{figure}

\subparagraph{A continuous graph $\Gr$.} We will instead define a
slightly different graph \Gr that includes additional ``Steiner''
points on $Q$. This \emph{does} allow us to accurately estimate the
distance to \emph{any} point on $Q$, and thus compute appropriate
anchor points. We actually construct the cone graph $G$ with parameter
$\eps_1=\eps/9$, and then extend it into a \emph{continuous graph}
$\Gr = (V^Q, E^Q)$ (see below). We obtain $\Gr$ from $G$ and $Q$ as
follows:

\begin{itemize}
\item We include all vertices of $G$ and of $Q$ as vertices of
  $\Gr$.
\item For each vertex $v\in G$ we consider the cones in the cone
  family centered at $v$. For each cone boundary $\rho$, we consider
  the first intersection point $q$ of $\rho$ with $Q$; if the
  resulting line segment $\overline{vq}$ lies inside $P$ (i.e. the ray
  does not intersect the boundary of $P$ before reaching $q$), we add
  $q$ as a vertex, and $(v,q)$ as an edge. Let $X_Q(v)$ denote the set
  of vertices added by $v$ in this way. See
  Figure~\ref{fig:anchor_points}(a) for an illustration.
\item We add all edges of $G$, and
\item we connect each vertex in $V^Q$ on $Q$ to the next
  vertex in $V^Q$ in the order along $Q$.
\end{itemize}

Observe that $\Gr$ still has size $O(nk)$; as each vertex in $G$ adds
$O(k)$ additional vertices and edges, and $Q$ also has at most $O(n)$
vertices. It is straightforward to construct $\Gr$ from $G$ in
$O(nk\log n)$ time, using data structures for fixed-directional ray
shooting
queries~\cite{sarnak86planar_point_locat_using_persis_searc_trees}.

\begin{lemma}
  \label{lem:construct_graph}
  In $O(nk\log n)$ time we can construct a data structure of size
  $O(nk)$ so that for any point $s \in P$ we can compute the extra
  vertices $X_Q(s)$ in $O(k\log n)$ time.
  Hence, \Gr can be constructed in $O(nk\log n)$ time.
\end{lemma}

\begin{proof}
  For each of the $O(k)$ cone directions we rotate the plane so that the
  cone direction is vertical, and store a vertical ray shooting query
  structure on (the edges of)
  $P$~\cite{sarnak86planar_point_locat_using_persis_searc_trees}. We
  similarly store the edges of $Q$ so that we can test for intersection
  with a vertical line segment. 
  In total this takes $O(kn\log n)$ time, and $O(nk)$ space. For each
  vertex $v \in G$ we can now compute $X_Q(v)$ in $O(k\log n)$ time:
  we perform two ray shooting queries from $v$ to find the maximal
  line segment in $P$, and the first intersection of this segment with
  $Q$. Sorting the vertices along $Q$ takes $O(nk\log n)$ time in
  total, as for every vertex $v$ the order of the points in $X_Q(v)$
  along $Q$ is consistent with the order on the orientation of the
  cone boundaries. So we can just merge-sort these $O(n)$ sorted lists
  in $O(nk\log n)$ time. We can compute all edge lengths in $O(nk)$
  time in total.
\end{proof}

Next, we argue that $\Gr$ indeed allows us to
$\eps_1$-approximate the distances. We actually interpret
$\Gr$ as a \emph{continuous
  graph}~\cite{cabello25algor_distan_probl_contin_graph}. Every edge
$(u,v) \in E^Q$ is a line segment $\overline{uv}$ in $P$, so
we can extend the (graph) distance $\dist_{\Gr}$ to points in
the interior of edges. In particular, for points $p,q$ on the same
edge $\overline{uv}$ we define $\dist_{\Gr}(p,q) = \|p-q\|$
to simply be their Euclidean distance, and for
$p \in \overline{uv} \in E^Q$ and
$q \in \overline{wz} \in E^Q$ in the interior of different
edges we have
$\dist_{\Gr}(p,q) = \min_{u' \in \{u,v\}}\min_{w' \in
  \{w,z\}} \|p-u'\| + \dist_{\Gr}(u',w') + \|w'-q\|$. We are only interested in points in the interior of
edges of $Q$. Hence, we prove:

\begin{figure}[tb]
  \centering
  \includegraphics[page=2]{anchor_points}
  \qquad
  \includegraphics{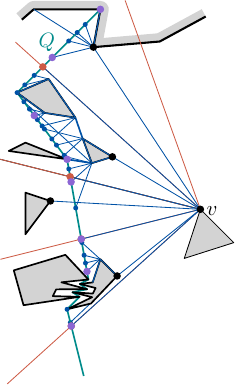}
  \caption{(a) The minimal cone neighbors $N(v)$ of vertex $v$, and
    the additional vertices $X_Q(v)$ (brown). (b) We compute a set of
    anchor points $A_Q(v)$ (purple) on $Q$ based on $\Gr$ (partially drawn).}
  \label{fig:anchor_points}
\end{figure}

\begin{lemma}
  \label{lem:vertex_distances}
  Let $u,v \in V$ be two polygon vertices in $\Gr$. Then we
  have that
    \begin{align*}
      \dist(u,v)\leq \dist_{\Gr}(u,v) \leq (1+\eps_1)\dist(u,v).
    \end{align*}
\end{lemma}
\begin{proof}
  By Lemma~\ref{lemm:Clarkson} the distance $\dist_G(u,v)$ between any
  pair of vertices $u,v \in G$ is at least $\dist(u,v)$ and at most
  $(1+\eps_1)\dist(u,v)$. Our construction adds only edges between
  pairs of (possibly new) vertices $(a,b)$ that are mutually visible,
  and thus have distance $\dist(a,b)$. It thus follows that for any
  pair of polygon vertices $u,v$ their distance in $\Gr$ is
  still at least $\dist(u,v)$ and at most $(1+\eps_1)\dist(u,v)$ as
  claimed.
\end{proof}

\begin{lemma}
  \label{lem:continuous-Q-graph_vertex_exists}
  Let $v$ be a vertex of $\Gr$, let $q \in Q$ be a point
  visible from $v$, and let $C$ be the cone (from the cone family
  whose cones have $v$ as apex) that contains $q$. There is a neighbor
  $u$ of $v$ in $\Gr$ that: (i) lies in $C$, and (ii) has
  (cone) distance less than or equal to that of $q$.
\end{lemma}

\begin{proof}
  Since $q$ lies on $Q$ it must lie on a line segment between two
  vertices $u_0$ and $u_1$ of the polygon. Without loss of generality
  assume that $u_0$ has cone distance less than or equal to that of
  $q$. See Figure~\ref{fig:neighbor_exists} for illustration.

  If $u_0$ lies in $C$ and is visible from $v$ then the lemma
  statement holds as $v$ must have a minimal cone neighbor in $C$
  whose cone distance is at most that of $u_0$.

  \begin{figure}[tb]
    \centering
    \includegraphics{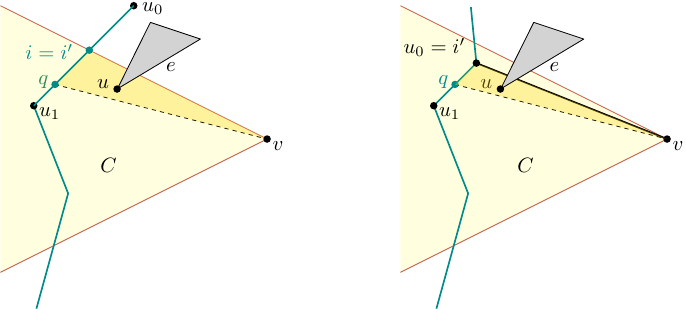}
    \caption{Point $q$ must lie on some edge $\overline{u_0u_1}$ of
      $Q$. We argue that there is a neighbor $u$ of $v$ with smaller
      cone distance than $q$.  }
    \label{fig:neighbor_exists}
  \end{figure}

  If on the contrary $u_0$ lies outside of $C$, then we consider the
  intersection $i$ between the line segment $\overline{u_0q}$ and one
  of the rays $\rho$ bounding the cone $C$. If this intersection point
  is visible then $i$ is a vertex in $\Gr$, $v$ has an edge
  to $i$ and $i$ is closer than $q$, so the lemma statement
  holds. Assume on the contrary that there exists an edge of the
  polygon that blocks visibility between $v$ and $i$. In particular,
  let $e$ be the first edge of $P$ intersected by $\rho$ (i.e. closest
  to $v$), and let $i' :=i$.

  Similarly, if $u_0$ does lie inside $C$ but is invisible, then there
  must be a first edge $e$ intersected by the line segment between $v$
  and $i' := u_0$.

  Per definition the line segments $\overline{qv}$ and
  $\overline{i'q}$ are not intersected by any edges of the polygon. So
  $e$ must have a visible endpoint inside of the triangle with corners
  $i'$, $v$, and $q$ (this is similar to
  Observation~\ref{obs:triangle_third_intersection}). This endpoint is
  a cone neighbor of $v$, and has cone distance at most that of
  $q$. The lemma statement now follows.
\end{proof}

\begin{lemma}
  \label{lem:continous-Q-graph-distance}
  For any point $q \in Q$ and any vertex $v \in
  \Gr$ we have that
    \begin{align*}
      \dist(v,q)\leq \dist_{\Gr}(v,q) \leq (1+\eps_1)\dist(v,q).
    \end{align*}
\end{lemma}

\begin{proof}
  Every edge of $\Gr$ corresponds to a line segment that lies
  inside $P$. Hence, any path in $\Gr$ is obstacle
  avoiding. It then also follows that for any pair
  $q, v \in \Gr$ we have
  $\dist(v,q)\leq \dist_{\Gr}(v,q)$.

  We prove the upper bound on $\dist_{\Gr}(v,q)$ using
  induction on the true distance to $q$. Let $v\in V^Q$ be a
  vertex and assume that the lemma statement holds for all vertices in
  $\Gr$ with a strictly smaller distance to $q$ than
  $\dist(v,q)$. For our base case there are no such vertices which
  means our assumption holds. If $v$ lies on $Q$ then the statement is
  trivially true. Otherwise, consider the shortest path from $v$ to
  $q$ in $P$. This path either passes through a vertex $w$ of the
  polygon or it is a line segment $\overline{qv}$.

  In the first case we use Lemma~\ref{lem:vertex_distances} to get
  that $\dist_{\Gr}(v,w) \leq (1+\eps_1)\dist(v,w)$ and the
  induction hypothesis (on the distance between $w$ and $q$) to
  obtain:
  \begin{align*}
    \dist_{\Gr}(v,q) &\leq \dist_{\Gr}(v,w)+\dist_{\Gr}(w,q)\\
                             &\leq (1+\eps_1)\dist(v,w) + (1+\eps_1)\dist(w,q)\\
                             &\leq (1+\eps_1)\dist(v,w) + (1+\eps_1)\dist(w,q)\\
                             &= (1 + \eps_1)\dist(v,q).
  \end{align*}

  In the second case the shortest path from $v$ to $q$ is a line
  segment, so consider the cone $C$ with apex at $v$ in which $q$ is
  visible. By Lemma~\ref{lem:continuous-Q-graph_vertex_exists} there
  exists a neighbor $u$ of $v$ in $\Gr$ that (i) $u$ also
  lies in cone $C$, and (ii) has (cone) distance less than or equal to
  that of $q$, in particular (iii) the neighbor $u$ has minimum cone distance of all vertices of $P$. This allows us to use Lemma~\ref{lem:path_via_cone_neighbor}, conclusion $(1)$ gives us that $\dist(u,q) \leq d(v,q)$. This allows us to use the induction hypothesis to obtain that
  $\dist_{\Gr}(u,q) \leq (1 + \eps_1)\dist(u,q)$. From conclusion $(2)$ of Lemma~\ref{lem:path_via_cone_neighbor}, it additionally follows that
  $\dist_{\Gr}(v,q) \leq (1 + \eps_1)\dist(v,q)$.
\end{proof}

\subparagraph{Computing and storing anchor points.} We can now compute
the distances in $\Gr$, and use them to construct an
appropriate set of anchor points. See
Figure~\ref{fig:anchor_points}(b). Due to
Lemma~\ref{lem:continous-Q-graph-distance} we can now actually just reuse
Thorup's approach (i.e. Lemma~\ref{lem:thorup_anchor_points}) again:


\begin{lemma}
  \label{lem:anchor_points_for_vertices}
  For every vertex $v$ of $P$, we can compute a set $A_Q(v)$ of $O(k)$
  anchor points on $Q$ such that for any point $q$ on $Q$ we have
  \begin{align*}
    \dist(v,q) \leq \dist_{G^Q}(v,A_Q(v),q) \leq (1 + \eps)\dist(v,q).
  \end{align*}
  This takes $O(nk^2\log n\log(kn))$ time in total.
\end{lemma}

\begin{proof}
  We construct $\Gr$ using Lemma~\ref{lem:construct_graph} in
  $O(nk\log n)$ time. We then apply
  Lemma~\ref{lem:thorup_anchor_points} with $1/k'=\eps_0 :=\eps/9$ to
  graph $\Gr$ to construct a set of points $A_Q(v)$ for each
  vertex $v$. This thus takes
  $O(k\log n(nk\log(nk) + nk))=O(nk^2\log n\log(nk))$ time. What
  remains is to argue that this set $A_Q(v)$ is an actual set of
  anchor points, i.e. that we can get an $(1+\eps)$-approximation of the
  distance $\dist(v,q)$ for any point $q$ on $Q$.

  \begin{figure}[tb]
    \centering
    \includegraphics{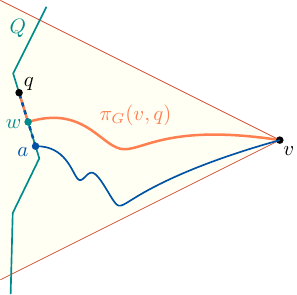}
    \caption{There is a sufficiently short path from $v$ to $q$ using
      anchor point $a \in A_Q(v)$.}
    \label{fig:anchor_points_vertices}
  \end{figure}

  Observe that for every point $q$ on $Q$ there is some vertex $w$ of
  $\Gr$ on $Q$ so that the shortest path from $q$ to $v$ in $\Gr$
  passes through $w$ (possibly $w=q$). By
  Lemma~\ref{lem:thorup_anchor_points}, there is an anchor point
  $a \in A_Q(v)$ for which
  $\dist_A(w,v) = \dist(q,a) + \hat{\dist}_a \leq (1 +
  \eps_0)\dist_{\Gr}(v,w)$. See
  Figure~\ref{fig:anchor_points_vertices} for an illustration. Since
  we can (at least) use the anchor point that is used by $w$ we have
  \begin{align*}
    \dist_{\Gr}(v,A_Q(v),q)&\leq \dist_{\Gr}(v,A_Q(v),w) + \dist(w,q)\\
                 &\leq (1 + \eps_0)\dist_{\Gr}(v,w)+ \dist(w,q).
  \end{align*}
  Then, using that $w$ lies on the shortest path from $v$ to $q$ in
  $\Gr$ (and that $\eps_0 > 0$) it follows that
  \begin{align*}
    \dist_A(v,q)&\leq (1 + \eps_0)\dist_{\Gr}(v,w)+\dist(w,q) \\
                 &\leq (1 + \eps_0)\dist_{\Gr}(v,q)
  \end{align*}
  Finally, Lemma~\ref{lem:continous-Q-graph-distance} gives us that
  the distance $\dist_{\Gr}(q,v)$ is actually a
  $(1+\eps_1)$-approximation of the distance between $q$ and $v$ in
  $P$. Using that $\eps_1 = \eps/9$, together with some basic math, we
  then indeed get:
  \begin{align*}
    \dist_A(v,q)&\leq (1 + \eps_0)\dist_{\Gr}(v,q)\\
                 &\leq (1+\eps_0)(1+\eps_1)\dist(v,q)\\
                 &\leq (1+\eps/9)^2\dist(v,q)\\
                 &\leq (1 + \eps)\dist(v,q).
  \end{align*}

  The lower bound again directly follows as paths in $\Gr$
  correspond to obstacle avoiding paths in $P$. This completes the
  proof.
\end{proof}


\subsection{The augmented cone graph}
\label{sub:The_augmented_cone_graph}

Let $S$ be a dynamic set of point sites inside $P$. As we are
interested in distances to (and between) the points in $S$, we also
want to define and compute anchor points on $Q$ for the points in
$S$. The most straightforward approach would be to just build the
above continuous graph $\Gr$ on the set of points (vertices)
$V \cup S$ rather than just the polygon vertices $V$. However, as
computing the anchor points involves computing shortest paths in this
graph, we cannot easily compute the anchor points for a single new
site $s$. Similarly, maintaining (the distances in) the graph as $S$
changes will be difficult. So, we argue that we can compute suitable
anchor points for these new sites using a slightly different method
instead.

Let $s \in P$ be a new point site. Similar to
Section~\ref{sub:Clarkson's_Cone_graph}, we now consider the
(continuous) graph $\Gr[s]$ that we obtain from
$\Gr$ by introducing $\{s\} \cup X_Q(s)$ as additional
vertices, connecting $s$ to every minimal cone neighbor in $N(s)$ and
to every point in $X_Q(s)$, and subdividing the edges of $Q$ to
account for the additional vertices in $X_Q(s)$. This (continuous)
graph still has $O(nk)$ vertices and $O(nk)$ edges. Moreover, we can
use this graph to accurately estimate distances from $s$ to any point
on $Q$.

\begin{lemma}
  \label{lem:augmented-cone-graph}
  Let $s$ be a point in $P$. For any vertex $w$ of $P$ there is a path from $s$ to $w$ in
  $\Gr[s]$ of length
  \[
    \dist(s,w) \leq \dist_{\Gr[s]}(s,w) \leq (1+\eps_1)\dist(s,w)
  \]
  that passes through a cone neighbor $v \in N(s)$.
\end{lemma}

\begin{proof}
  This follows directly from
  Lemma~\ref{lem:path_via_cone_neighbor_exists} since
  $\Gr[s]$ has $G[s]$ as an subgraph.
\end{proof}

\begin{lemma}
  \label{lem:distances_in_augmented_graph}
  Let $s$ be any point in $P$. For any point $q \in Q$ we have
  \[
    \dist(s,q) \leq \dist_{\Gr[s]}(s,q) \leq
    (1+\eps_1)\dist(s,q).
  \]
\end{lemma}

\begin{proof}
  The shortest path $\geod(s,q)$ from $s$ to $q$ either passes through
  a polygon vertex $w$ or is a line segment $\overline{sq}$.

  In the former case, $w$ is also a vertex in $\Gr$, and thus
  by Lemma~\ref{lem:continous-Q-graph-distance} the value
  $\dist_{\Gr}(w, q)$ is a $\eps_1$-approximation of the
  distance $\dist(w, q)$. By Lemma~\ref{lem:augmented-cone-graph} there
  is a path from $s$ to $w$ in $\Gr[s]$ via a minimum cone
  neighbor $v \in N(s)$ of $s$ of length $(1+\eps_1)\dist(s,w)$. We
  thus have
  \begin{align*}
    \dist_{\Gr[s]}(s,q) &\leq  \dist_{\Gr[s]}(s, w) + \dist_{\Gr[s]}(w, q)\\
                               &\leq \dist_{\Gr[s]}(s,w) + \dist_{\Gr}(w,q)\\
                               &\leq (1+\eps_1)\dist(s,w)   + (1+\eps_1)\dist(w,q)\\
                               &= (1+\eps_1)(\dist(s,w) + \dist(w,q))
                                 = (1+\eps_1)\dist(w,q).
  \end{align*}

  In the latter case, the exact same argument as in
  Lemma~\ref{lem:continuous-Q-graph_vertex_exists} (for now for the
  graph $\Gr[s]$) gives us that there exists a vertex
  $u \in \Gr[s]$ in the cone $C$ containing $q$ that (i) has
  cone distance at most $d^C(s,q)$, and (ii) that is connected by an
  edge to $s$. As we argue next, there is a path in $\Gr[s]$
  from $u$ to $q$ of length at most
  $(1+\eps_1)\dist(u,q)$. Lemma~\ref{lem:path_via_cone_neighbor} then
  gives us that the path from $s$ to $q$ in $\Gr[s]$ via $u$
  also has length at most $(1+\eps_1)\dist(s,q)$.

  If $u$ is also a vertex of $\Gr$, then by
  Lemma~\ref{lem:continous-Q-graph-distance} there then is a path in
  $\Gr$ from $u$ to $q$ of length at most
  $(1+\eps_1)\dist(u,q)$. Clearly, this path also exists in
  $\Gr[s]$. If $u$ is not a vertex of $\Gr$ then it
  lies on $Q$ (as it is the intersection of $Q$ with a boundary of
  $C$), and thus the subpath along $Q$ (which is also a path in
  $\Gr[s]$) has length
  $\dist(u,q) \leq (1+\eps_1)\dist(u,q)$.

  Note that once again, in both cases we actually have a path in
  $\Gr[s]$ which corresponds to an obstacle avoiding path in
  $P$. Hence, $\dist_{\Gr[s]}(s,q)$ is also at least $\dist(s,q)$.
\end{proof}

\subparagraph{Defining anchor points.} Our goal is to define a set
$A_Q(s)$ of anchor points on $Q$ for point $s$. To this end, we first
define a set of candidate points
$X'_Q(s) = X_Q(s) \cup \bigcup_{v \in N(s)} A_Q(v)$ on $Q$. We can
prove that if we assign appropriate weights to this set of $O(k^2)$
points it is indeed a set of anchor points. However, as we argue next,
a subset of $O(k)$ of these points suffices.

Let $H = (V',E')$ be an undirected weighted graph with vertex set
$V' = \{s\} \cup X'_Q(s)$. Vertex $s$ is connected to each vertex
$p \in X_Q(s)$ with an edge of length
$\|s-p\| = \dist_{\Gr[s]}(s,p)$, and to $p \in A_Q(v)$ by an
edge of length $\|s-v\|+\dist_{\Gr}(v,p)$. We connect
vertices $u,w \in X'_Q(s)$ that appear consecutively on $Q$ by an edge
of length $\dist(u,w)$ (which is the length of the subpath along
$Q$). Observe that the vertices of $X'_Q(s)$ form a shortest path in
the graph $H$.

\newcommand{\filterloss}{\ensuremath{\varepsilon_2}}
\begin{lemma}
  \label{lem:prune_connections}
  Let $\filterloss > 0$ be a parameter such that $O(1/\filterloss) = O(k)$. In $O(k^2)$ time we can
  compute a set $A_Q(s)$ of $O(k)$ discrete anchor points for $H$,
  i.e. so that for any $p \in X'_Q(s)$
  \[
    \dist_H(s,p) \leq \dist_H(s,A_Q(s),p) \leq (1 + \filterloss)\dist_H(s,p).
  \]
\end{lemma}
\newcommand{\weight}{\ensuremath{\mathit{weight}}\xspace}

\begin{proof}
  This proof follows the idea of Lemma 5 of Thorup~\cite{Thorup2007}. Let $Q$ be the shortest path $\pi(u,v)$, we define $\weight(p)$ as the edge length from $s$ to $p$, in particular for all $p$ on $Q$ we have that $\weight(p) \geq \dist_H(s, p)$. Let $q^* = \argmin_p \weight(p)$, observe that $\weight(q^*) = \dist_H(s, q^*)$. As all $p,q\in Q$ lie on a shortest path $Q$ we can always find $\dist_H(p, q)$ in constant time. We consider the points on $Q$ ordered by their distance $\dist_H(p, v)$ to the endpoint of $Q$. In particular, we will say $b$ occurs \textit{after} $a$ on $Q$ if $\dist_H(b, v) \leq \dist_H(a, v)$.

  Our first anchor point $a_0$ is $q^*$. We will explain how to find
  all anchor points which occur after $q^*$ on $Q$, finding all anchor
  points which occur before $q^*$ is symmetrical to this. As
  preprocessing for our algorithm we filter $X'_Q(s)$ by removing all
  elements $p$ which either occur before $q^*$ or if $\weight(p) >
  \dist_H(s, q^*)+ \dist_H(q^*, p)$. This preprocessing, including
  finding $q^*$ itself, takes $O(|X'_Q(s)|) = O(k^2)$ time. Recall
  that $X'_Q(s)$ consists of $O(k)$ lists containing $O(k)$ elements
  each, which are all ordered by their position on $Q$. We
  can generate the set $X_Q(s)$ in order around $s$, and this order is
  consistent with the order on $Q$.
  Our algorithm to find the anchor points performs the following steps until all these lists are empty.
  \begin{itemize}
    \item In round $j$ we iterate over each of the ordered lists $i$
      until we find the first element $b_i$ for which $(1+\eps_2)\weight(b_i) <
      \weight(a_{j - 1}) + \dist_H(a_{j - 1}, b_i)$.
    \item For each list $i$ we delete all elements we traversed over for which the above condition did not hold. In other words, all elements $p$ which occur before $b_i$ on $Q$ for which the previous anchor point $a_{j - 1}$ covers $p$ are deleted from their respective lists.
    \item We define $a_j$ as the $b_i$ which occurs the earliest on $Q$, this element is also removed from their respective list.
  \end{itemize}
  The time complexity of each round of the above algorithm is $O(k + T)$ where $T$ is the number of elements deleted. Since each element can be deleted at most once, in total we spend $O(kL + k^2)$ time. Where $L$ is the number of anchor points we find.

  We will now prove that $L = O(k)$ (which is essentially the same
  argument as in Lemma 5 of Thorup~\cite{Thorup2007}). For this we define a function $f(a) = \weight(a) + \dist_H(a, v)$. We observe that for every anchor point $a_j$ with $j \geq 1$ we have per definition of $f$ and $a_j$ that
  \begin{align*}
    f(a_{j - 1}) - f(a_j) &= \weight(a_{j - 1}) + \dist_H(a_{j - 1}, a_j) - \weight(a_j)\\
    &> \filterloss \cdot \weight(a_j)\\
    &\geq \filterloss \cdot \dist_H(s, q^*).
  \end{align*}
  So $f$ decreases by at least $\filterloss \cdot \dist_H(s, q^*)$ for each anchor point. Furthermore, we have per triangle equality that for all anchor points $a_j$
  \begin{align*}
    \dist_H(s, v) &\leq \dist_H(s, a_j) + \dist_H(a_j, v)\\
    &\leq \weight(a_j) + \dist_H(a_j, v) = f(a_j)
  \end{align*}
  and since $a_j$ was not filtered out in the preprocessing, we have that
  \begin{align*}
    f(a_j) = \weight(a_j) + \dist_H(a_j, v)
    &\leq \dist_H(s, q^*) + \dist_H(q^*, a_j) + \dist_H(a_j, v)\\
    &=  \dist_H(s, q^*) + \dist_H(q^*, v)    \\
    &\leq \dist_H(s, q^*) + (\dist_H(q^*, s) + \dist_H(s, v))\\
    &\leq \dist_H(s,v) + 2\dist_H(s, q^*).
  \end{align*}

  Since the difference between the upper and lower bound is $2\dist_H(s, q^*)$ we find that there can be at most $2/\filterloss$ anchor points after $q^*$. By performing the above algorithm and proof symmetrically for all anchor points before $q^*$ we can construct $A_Q(s)$, consisting of $O(1/\filterloss) = O(k)$ anchor points in $O(k^2)$ time.

  Finally, we will show that $A_Q(s)$ is a valid set of anchor points on
  the graph $H$. The lower bound easily holds as the path via any
  anchor point is a valid path in $H$. To show the upper bound we need
  to prove that for any vertex $p \in X'_Q(s)$ there exists an $a\in A_Q(s)$ such that
  \begin{align*}
    \weight(a) + \dist_H(a, p) \leq (1 + \filterloss)\dist_H(s, p)
  \end{align*}
  We separate into two cases
  \begin{itemize}
    \item[(i)] If $\dist_H(s, p) = \weight(p)$, then per triangle inequality $\dist_H(s, p) \leq \dist_H(s, q^*) + \dist_H(q^*, p)$ we find that the point $p$ was not filtered out in the preprocessing step. This means that either $p$ is an anchor point itself or there exists an anchor point $a\in A_Q(s)$ such that $\weight(a) + \dist_H(a, p) \leq (1 + \filterloss)\dist_H(s, p)$, and thus the upper bound holds.
    \item[(ii)] If $\dist_H(s, p) < \weight(p)$ then the shortest path $\pi_H(s,p)$ must pass through other vertices. Let $b$ be the first vertex after $s$ on this path. We now have that $\dist_H(s, b) = \weight(b)$. Per case $(i)$ there must now exist a vertex $a$ such that
    \begin{align*}
      \weight(a) + \dist_H(a, p) &\leq \weight(a) + \dist_H(a, b) + \dist_H(b, p)\\
      &\leq (1 + \filterloss)\dist_H(b) + (1 + \filterloss)\dist_H(b, p)\\
      &= (1 + \filterloss)\dist_H(p),
    \end{align*}
    which proves the upper bound. \qedhere
  \end{itemize}
\end{proof}

We then use Lemma~\ref{lem:prune_connections} with $\filterloss:=\eps/9$ to compute a set $A_Q(s)$ of
$O(k)$ discrete anchor points for $s$ with respect to graph
$H$. Hence, for each point $p \in X'_Q(s)$ (so in particular for
those in a set $A_Q(v)$) there is a path of length
$(1+\eps_2)\dist_{H}(s,p)\leq(1+\eps_2)(\|s-v\|+\dist_{\Gr}(v,p))$
via an anchor point in $A_Q(s)$. Next, we argue that this set is a valid set
of anchor points for any point on $Q$.

\begin{lemma}
  \label{lem:query_anchor_set_is_valid}
  The set $A_Q(s)$ is a valid set of anchor points, so for any point $q
  \in Q$ we have
  \[
    \dist(s,q) \leq \dist_{\Gr[s]}(s,A_Q(s), q) \leq (1+\eps)\dist(s,q).
  \]
\end{lemma}

\begin{proof}
  The argument is similar to that of
  Lemma~\ref{lem:anchor_points_for_vertices}. For any point $q \in Q$,
  the shortest path from $q$ to $s$ in $\Gr[s]$ passes
  through some vertex $w$ of $\Gr[s]$ on $Q$.

  If $w$ is a point in $X_Q(s)$, then we have
  $\dist_{\Gr[s]}(s,A_Q(s),q) \leq \hat{\dist_w} + \dist(w,q) =
  \dist_{\Gr[s]}(s,q)$, which, by
  Lemma~\ref{lem:distances_in_augmented_graph} is at most
  $(1+\eps_1)\dist(s,q)$. Since $\eps_1 \leq \eps$ we then also have
  $\dist_{\Gr[s]}(s,X'_Q(s), q) \leq (1+\eps)\dist(s,q)$.

  Otherwise, $w$ must also be a vertex of $\Gr$, and the
  shortest path from $s$ to $w$ in $\Gr[s]$ passes through
  some cone neighbor $v \in N(s)$. Let $p \in A_Q(v)$ be the anchor
  point used in $\Gr$ that connects $v$ to $w$; i.e. so that
  (by Lemma~\ref{lem:thorup_anchor_points}) we have that
  $\dist_{\Gr}(v,A_Q(v),w) = \dist_{\Gr}(v,p) + \dist(p,w) \leq
  (1+\eps_0)\dist_{\Gr}(w,v) = \dist_{\Gr}(w,v) +
  \eps_0\dist_{\Gr}(w,v)$.

  Furthermore, let $a \in \subseteq A_Q(s)$ be the anchor point
  used in $H$ to get from $s$ to $p$; i.e. so that there is a path
  via $a$ of length at most
  $\dist_{\Gr[s]}(s,A_Q(s),a)+\dist(a,p) \leq (1+\eps_2)\dist_{H}(s,p) \leq
  (1+\eps_2)(\|s-v\|+\dist_{\Gr}(v,p))$. Since we can use
  anchor point $a$ to get from $s$ to $q$ via $p$ and $w$, using Lemma~\ref{lem:prune_connections} we have
  \begin{align*}
    \dist_{\Gr[s]}(s,A_Q(s),q) &\leq \dist_{\Gr[s]}(s,A_Q(s),a) + \dist(a,q)\joost{I believe we can already omit $A_Q(s)$ here} \\
                   &\leq \dist_{\Gr[s]}(s,A_Q(s), a)  + \dist(a,p) + \dist(p,w)  + \dist(w,q)\\
                    &\leq \eps_2(\|s-v\|+\dist_{\Gr}(v,p))
                      + \|s-v\|+\dist_{\Gr}(v,p) + \dist(p,w) + \dist(w,q)\\
                    &\leq \eps_2(\|s-v\|+\dist_{\Gr}(v,p))
                      + \|s-v\|+\eps_0\dist_{\Gr}(v,w) +
                      \dist_{\Gr}(v,w) + \dist(w,q)\\
                    &\leq \eps_2(\|s-v\|+\dist_{\Gr}(v,p))
                      + \eps_0\dist_{\Gr}(v,w) +
                         \dist_{\Gr[s]}(s,w) + \dist(w,q)\\
                    &\leq \eps_2(\|s-v\|+\dist_{\Gr}(v,p))
                      + \eps_0\dist_{\Gr}(v,w)
                      + \dist_{\Gr[s]}(s,q)
  \end{align*}
  Observe that
  $\|s-v\|+\dist_{\Gr}(v,p) \leq
  \|s-v\|+(1+\eps_0)\dist_{\Gr}(v,w) \leq
  (1+\eps_0)\dist_{\Gr[s]}(s,w) \leq
  (1+\eps_0)\dist_{\Gr[s]}(s,q) $.  Furthermore, by
  Lemma~\ref{lem:distances_in_augmented_graph} the distance from $s$
  to $q$ in $\Gr[s]$ is at most
  $(1+\eps_1)\dist(s,q)$. Plugging this in we thus get
  \begin{align*}
    \dist_{\Gr[s]}(s,A_Q(s),q)
                   &\leq \eps_2(\|s-v\|+\dist_{\Gr}(v,p))
                      + \eps_0\dist_{\Gr}(v,w)
                      + \dist_{\Gr[s]}(s,q) \\
                    &\leq \eps_2(1+\eps_0)\dist_{\Gr[s]}(s,q)
                        + \eps_0\dist_{\Gr[s]}(s,q)
                        + \dist_{\Gr[s]}(s,q) \\
                    &\leq (\eps_2(1+\eps_0)+\eps_0+1)\dist_{\Gr[s]}(s,q)\\
                    &\leq (\eps_2(1+\eps_0)+\eps_0+1)(1+\eps_1)\dist(s,q)\\
  \end{align*}

  Finally, using that
  $\eps_0 = \eps_1 = \eps_2 = \frac{\eps}{9} \leq 1$ we have
  $(\eps_2(1+\eps_0)+\eps_0+1)(1+\eps_1) \leq
  (\eps_2+\eps_2\eps_0+\eps_0+1)(1+\eps_1) \leq
  (\frac{3\eps}{9}+1)(1+\frac{\eps}{9}) \leq (1+\frac{\eps}{3})^2 \leq
  1+\eps$ and thus $\dist_{\Gr[s]}(s,A_Q(s),q) \leq (1+\eps)\dist(s,q)$ as
  desired.

  The lower bound again follows easily since $\dist_{\Gr[s]}(s,A_Q(s),q)$ is the
  length of an obstacle avoiding path from $\Gr[s]$.
\end{proof}



\begin{lemma}
  \label{lem:query_anchor_set}
  Let $s$ be a point in $P$. In $O(k\log n + k^2)$ time, we can
  compute a set $A_Q(s)$ of $O(k)$ anchor points, so that
  for any point $q \in Q$ we have
  \[
    \dist(s, q) \leq \dist_{\Gr[s]}(s,A_Q(s), q) \leq (1 + \eps)\dist(s,q).
  \]
\end{lemma}

\begin{proof}
  Correctness follows by Lemma~\ref{lem:query_anchor_set_is_valid}, so
  all that remains is to argue about the construction time. Observe
  that the set $X'_Q(s)$ has size $O(k^2)$ and can be constructed in
  $O(k\log n + k^2)$ time, as we can compute $X_Q(s)$ and $N(s)$
  in $O(k\log n)$ time using the results from
  Lemmas~\ref{lemm:Clarkson} and~\ref{lem:construct_graph}.

  For each of these points we also compute their weights (i.e. edge
  lengths in the graph $H$) in $O(k^2)$ time, and then use
  Lemma~\ref{lem:prune_connections} to compute the anchor points in
  additional $O(k^2)$ time. (Note that we do not need to explicitly
  construct the full graph $H$ to apply Lemma
  Lemma~\ref{lem:prune_connections}). The lemma follows.
\end{proof}





\subparagraph{Estimating distances via the shortest path $Q$.} Given a
shortest path $Q$, and points $s,t \in P$ whose shortest path
$\geod(s,t)$ in $P$ intersects $Q$, we can use the above tools to
accurately estimate the distance between $s$ and $t$. More
specifically, we compute $G^Q$ and the anchor points as defined
above. Additionally, we fix an endpoint $u$ of $Q$ and compute and
store, for each vertex $w$ of $Q$ the distance to $u$ (i.e. the path
length along $Q$). This can easily be done in $O(n)$ additional
time. Given a pair of points $a,b$ on $Q$ (and the edges that contain
$a$ and $b$, respectively) we can then compute their distance
$\dist(a,b)$ (i.e. the length of the subpath $Q[a,b]$) in constant
time.

Let $a \in A_Q(s)$ and $b \in A_Q(t)$ be anchor points of $s$ and
$t$ on $Q$, respectively. The subpath $\hat{\geod}_a$ from $s$ to $a$
in $\Gr[s]$ corresponds to an obstacle avoiding path in
$P$. Similarly, the subpath $\hat{\geod}_b$ from $b$ to $t$ is
obstacle avoiding. Furthermore, since $Q$ is a shortest
(obstacle-avoiding) path in $P$, so is its subpath $Q[a,b]$ from $a$
to $b$. So, the concatenation
of $\hat{\geod}_a$, $Q[a,b]$, and $\hat{\geod}_b$ connects $s$ and $t$
and is obstacle avoiding, and has length
$\hat{\dist}_a + \hat{\dist}_b + \dist(a,b)$. Now define
\[
  \anchorDistancePoints{s}{t} = \min_{a \in A_Q(s), b \in A_Q(t)} \hat{\dist}_a + \hat{\dist}_b + \dist(a,b)
\]
to be the length of a shortest such obstacle avoiding path $\hat{\geod}_Q(s,t)$. We then have:

\begin{lemma}
  \label{lem:graph_and_shortest_path_datastructure}
  Let $P$ be a polygonal domain with $n$ vertices, and let $Q$ be a
  shortest path in $P$. In $O(nk^2\log n\log(kn))$ time, we can
  construct a data structure on $P$ and $Q$ of size $O(nk)$ that can
  answer the following queries in $O(k\log n + k^2)$ time: given any
  pair of points $s,t \in P$ for which a shortest path $\geod(s,t)$
  intersects $Q$, return the length $\anchorDistancePoints{s}{t}$ of an
  obstacle avoiding path $\hat{\geod}_Q(s,t)$ so that
  \[
    \dist(s,t) \leq \anchorDistancePoints{s}{t} \leq (1+\eps)\dist(s,t).
  \]
\end{lemma}

\begin{proof}
  The bound on the space and preprocessing time follow directly from
  Lemma~\ref{lem:anchor_points_for_vertices}, so what remains is to
  argue that we handle queries within the stated bounds.

  Let $q$ be an intersection point of $\geod(s,t)$ and $Q$. By
  Lemma~\ref{lem:query_anchor_set} there are anchor points
  $a \in A_Q(s)$ and $b \in A_Q(t)$ for which
  $\dist_{\Gr[s]}(s,A_Q(s), q) = \hat{\dist}_a + \dist(a,q) \leq
  (1+\eps)\dist(s,q)$ and
  $\dist_{\Gr[t]}(t,A_Q(t), q) = \hat{\dist}_b + \dist(b,q) \leq
  (1+\eps)\dist(t,q)$. It then follows that the path $\hat{\geod}_Q(s,t)$ via anchors
  points $a$ and $b$ has length at most
  \begin{align*}
    \anchorDistancePoints{s}{t} &\leq \hat{\dist}_a + \dist(a,b) + \hat{\dist}_b \\
                 &\leq \hat{\dist}_a + \dist(a,q) + \dist(q,b) + \hat{\dist}_b \\
                 &\leq (1+\eps)\dist(s,q) + (1+\eps)\dist(q,t) \\
                 &= (1+\eps)\dist(s,t).
  \end{align*}
  As argued, this path $\hat{\geod}_Q(s,t)$ is obstacle avoiding, and
  thus has length at least $\dist(s,t)$. By
  Lemma~\ref{lem:query_anchor_set}, we can compute the sets of anchor
  points $A_Q(s)$ and $A_Q(t)$ in $O(k\log n + k^2)$ time. Each such
  a set has size $O(k)$, so we can easily compute
  $\anchorDistancePoints{s}{t}$ in additional $O(k^2)$ time.
\end{proof}

\section{An improved data structure for distance queries}
\label{sec:An_improved_data_structure_for_distance_queries}

In this section we develop an $O(nk\log n)$ space data structure that
stores the polygonal domain $P$, so that given a pair of query points
$s,t$ we can compute an $(1+\eps)$-approximation of the distance
$\dist(s,t)$ in $O(k^2\log n)$ time. Our data structure is essentially
a balanced hierarchical subdivision of the polygonal domain combined
with the graph-based approach to estimate distances from the previous
section.

\subparagraph{A balanced hierarchical subdivision for polygonal
  domains.} A balanced hierarchical
subdivision~\cite{chazelle1989balanced} is a recursive subdivision of
a simple polygon into subpolygons. It is a crucial ingredient to
obtain efficient solutions to many problems involving simple
polygons~\cite{agarwal18improv_dynam_geodes_neares_neigh,2PSP_simple_polygon}. As
was also recently observed~\cite{berg24spanners}, Thorup's
approach~\cite{Thorup2007} essentially provides a balanced
hierarchical subdivision for polygonal domains. The following lemma
summarizes this result:

\begin{figure}[tb]
  \centering
  \includegraphics[page=2]{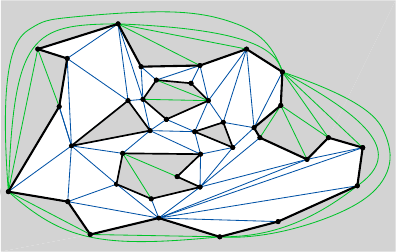}
  \quad
  \includegraphics[page=4]{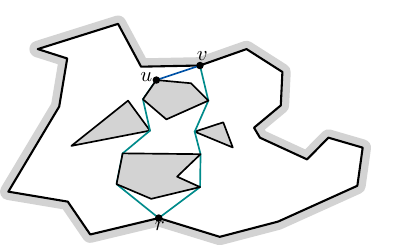}
  \caption{One step in the construction of the balanced hierarchical
    separator.}
  \label{fig:separator}
\end{figure}

\renewcommand{\Q}{\ensuremath{\mathcal{Q}}\xspace}
\begin{lemma}[\cite{Thorup2007,berg24spanners}]
  \label{lem:separator_subdivision}
  Let $P$ be a polygonal domain with $n$ vertices. In $O(n\log n)$
  time, one can compute a tree \T so that
  \begin{enumerate}
  \item the leaves of \T correspond to triangles in a triangulation of
    $P$,
  \item each node $\nu$ of $T$ corresponds to a subpolygon $P_\nu$,
  \item each internal node $\nu$ also corresponds to a set
  $\Q_\nu$ of shortest paths between vertices of $P$ so that:
    \begin{enumerate}
    \item $\Q_\nu$ consists of at most three shortest paths,
    \item every shortest path $Q \in \Q_\nu$ is also a shortest path
      in $P_\nu$, and
      \item the shortest paths in $\Q_\nu$ partition $P_\nu$ into
        in two subpolygons $P_\mu$ and $P_\omega$, corresponding to
        the children of $\nu$,
      \end{enumerate}
    \item the height of the tree is $O(\log n)$, and
    \item the total complexity of the subpolygons $P_\nu$ on each level of \T
      is $O(n)$, and thus the total complexity of \T is $O(n\log n)$.
    \end{enumerate}
\end{lemma}

\begin{proof}
  As stated, these results are implicitly used by
  Thorup~\cite{Thorup2007, Thorup2004}. As his version only claims
  that $\Q_\nu$ may consist of six shortest paths, and does not
  explicitly state the bounds on the complexities of the subpolygons
  $P_\nu$, we give an explicit proof here.

  Thorup's idea is as follows: we compute the shortest path tree $T$
  of some vertex $r$, and compute a triangulation of $P$ that includes
  the edges of $T$. We interpret the result as a triangulated planar
  graph $G$ by including additional edges so that all faces are
  triangles. That is, we triangulate the holes and make sure that the
  outer face is also a triangle. This may use additional edges that do
  not have a straight line embedding. See Figure~\ref{fig:separator}.

  We then use Lipton and Tarjan's fundamental cycle
  separator~\cite{Lipton-Endre}, which states that given a weighted
  triangulated planar graph $G$ and a spanning tree $T$ of $G$, one
  can find two vertices $u$ and $v$, so that the path connecting $u$
  and $w$ in $T$ and the missing edge $(u,v)$ form a cycle $C$ in $G$
  that partitions $G$ into three parts $A$, $B$, and $C$, and so that
  the weight in $A$ and $B$ is at most $2/3^\mathrm{th}$ of the total
  weight.

  Observe that since $T$ is a shortest path tree, the cycle $C$ consists
  of at most three shortest paths: the two paths connecting $u$ and
  $v$ to their lowest common ancestor in $T$, and potentially the line
  segment $\overline{uv}$, if the edge $(u,v)$ of the triangulation
  actually lies inside $P$.

  We now recursively apply this cycle separator to recursively
  partition the triangles of $G$. That is, we split $G$ into two
  subgraphs $G_1$ and $G_2$, so that a triangle of $G$ appears in
  either $G_1$ or $G_2$, and so that the number of triangles for both
  graphs $G_i$ is at most $2N/3$, where $N$ is the number of triangles
  in $G$. The cycle also splits $T$ into two spanning trees $T_1$ and
  $T_2$ (which are also still shortest path trees in $G_1$ and $G_2$,
  respectively).

  Initially, $G$ consists of $N = O(n)$ triangles, hence after
  $O(\log N) = O(\log n)$ rounds of recursion, every subgraph is a
  single triangle.

  Our tree \T will be the recursion tree of this process after
  pruning subtrees in which all triangles in the remaining graph lie
  outside of $P$. The tree thus indeed has height $O(\log n)$
  (property 4). Let $G_\nu$ be the remaining subgraph that we have at
  node $\nu \in \T$. For every leaf $\nu$, let $P_\nu$ be the single
  triangle in $G_\nu$. For every internal node $\nu$ we define $P_\nu$
  to be the union of the triangles in its subtree. This establishes
  properties (1) and (2). Furthermore, define $\Q_\nu$ to be the set
  of at most three shortest paths used by the cycle that we
  constructed to split $G_\nu$ (this establishes properties 3a and
  3c). Observe that these paths are paths in $P_\nu$, and as they are
  shortest paths in $P$ thus also shortest paths in $P_\nu$ (property
  3b).

  Since $P_\nu$ is the union of the triangles of $G_\nu$, the
  complexity of $P_\nu$ is $O(N_\nu)$, where $N_\nu$ is the number of
  triangles in $G_\nu$. At any level of the recursion tree, the total
  number of triangles is at most $N$ (as every triangle appears in
  only one subproblem), and thus the total complexity of the polygons
  on a level is also at most $O(N)=O(n)$. Since the tree has height
  $O(\log n)$, the total complexity is thus indeed $O(n\log n)$. This
  establishes property (5).

  Finally, note that we can compute $T$ in $O(n\log n)$
  time~\cite{HershbergerSuri1999} (or even in $O(n + h\log h)$ time if
  we have an initial triangulation of
  $P$~\cite{wang23new_algor_euclid_short_paths_plane}). Computing the
  cycle separator takes $O(n)$ time~\cite{Lipton-Endre}, and thus the total time to
  compute these separators (and construct the appropriate subpolygons
  and shortest paths) is also $O(n\log n)$.
\end{proof}

\subparagraph{The data structure.} Our data structure now stores the
balanced hierarchical subdivision \T from
Lemma~\ref{lem:separator_subdivision}. We build a point location data
structure~\cite{sarnak86planar_point_locat_using_persis_searc_trees}
on the triangulation induced by the leaves of \T. Additionally, for
each edge $e$ of this triangulation, let $\nu$ be the highest node in
\T for which the edge appears on a shortest path $Q \in \Q_\nu$ (if
such a node exists). We store a pointer from edge $e$ to this shortest
path. We can easily compute these pointers during the construction of
\T.

Consider a node $\nu$. For each of the $O(1)$ shortest paths
$Q \in \Q_\nu$ we will construct the data structure from
Lemma~\ref{lem:graph_and_shortest_path_datastructure} (i.e. the graph
$G^Q$ and the sets of anchor points) with respect to the subpolygon
$P_\nu$. 
Since the total size of all subpolygons $P_\nu$ on each level of the
tree is $O(n)$, the total size of all of these data structures on a
given level of \T is $O(kn)$, and thus $O(nk\log n)$ in
total. Similarly, the total preprocessing time is
$O(nk^2\log^2 n\log(kn))$.







\subparagraph{Query.} Given a pair of query points $s,t$ we find the leaf
triangles $\Delta_s$ and $\Delta_t$ that contain points $s$ and $t$,
respectively. If $s$
and $t$ lie in the same triangle, we can trivially report the length
$\overline{st}$ as their shortest path. Otherwise, a shortest path
from $s$ to $t$ must intersect one of the shortest paths
$Q \in \Q_\nu$ for a common ancestor $\nu$ of the two leaves. In
particular, let $\mu$ be the highest node of \T for which $\geod(s,t)$
intersects a shortest path $Q \in \Q_\mu$. It follows that the
shortest path $\geod(s,t)$ lies inside $P_\mu$, and thus we can use the
Lemma~\ref{lem:graph_and_shortest_path_datastructure} data structure for $Q$
to obtain a path of
length at most $(1+\eps)\dist(s,t)$. Unfortunately, we cannot easily
find this node $\mu$, so we query all
Lemma~\ref{lem:graph_and_shortest_path_datastructure} data structures
for all shortest paths $Q \in \Q_\nu$, for all common ancestors $\nu$
of the leaves representing $\Delta_s$ and $\Delta_t$, and report the
shortest distance found. If either $s$ or $t$ lies on an edge $e$ of
the boundary of some triangle that is part of some shortest path
$Q \in Q_\nu$, we use the pointer stored at $e$ to jump to the highest
such node, and query only the common ancestors of these nodes
instead. Since every query returns the length of an obstacle-avoiding
path, the returned estimate $\hat{\dist}(s,t)$ satisfies
$\dist(s,t)\leq \hat{\dist}(s,t) \leq (1+\eps)\dist(s,t)$ as
desired. Since each query takes $O(k\log n + k^2)$ time, the total
query time is $O(k\log^2 n + k^2\log n)$.

\subparagraph{Improving the query time.} The $O(k\log^2 n)$ term in
the query time is only due to computing the ray shooting queries
needed to compute the sets $X_Q(s)$ for each shortest path $Q$
considered. Next, we argue that we can compute these sets in only
$O(k\log n)$ time by building an additional $O(nk\log n)$ space data
structure. The overall space usage thus remains the same, and the
final query time becomes $O(k^2\log n)$ as claimed.

\begin{lemma}
  \label{lem:improved_query_time}
  In $O(nk\log n)$ time, we can build an $O(nk\log n)$ space data
  structure that given a query point $s$, allows us to compute the
  sets $X_Q(s)$ for all shortest paths
  $Q \in \bigcup_{\nu \in \T(s)} \Q_\nu$, where $\T(s)$ is the path in
  $\T$ from the root to the leaf triangle containing $s$.
\end{lemma}

\begin{proof}
  Fix the direction of one of the bounding lines of a cone, and assume
  without loss of generality the direction is vertically upward. We
  build one additional data structure, that given a query point $s$,
  allows us to compute the set $Y(s)$ of all $O(\log n)$ intersection
  points between the upward vertical ray from $s$ and all shortest
  paths $Q \in \Q_\nu$ of nodes $\nu$ on the root-to-leaf path $\T(s)$. In particular, we will answer such
  a query in $O(\log n)$ time, using $O(n\log n)$ space. We build such
  a structure for all $O(k)$ cone directions, leading to a total of
  $O(kn\log n)$ additional space. The total time to compute all
  relevant sets $X_Q(s)$ will be $O(k\log n)$.

  Rather than storing every shortest path $Q$ in an individual data
  structure that allows us to test for intersections of $Q$ with a
  vertical segment, we build a single structure on all edges that appear
  in shortest paths over all nodes of \T. For each such an edge $e$, let
  $\nu$ be the highest node in \T for which $e$ is an edge of a shortest
  path $Q \in \Q_\nu$, and let $\ell_e := \ell_\nu$ be the level of this
  node $\nu$ in \T. We interpret $\ell_e$ as a color. See
  Figure~\ref{fig:colored_ray_shooting} for an illustration.

  \begin{figure}[tb]
    \centering
    \includegraphics{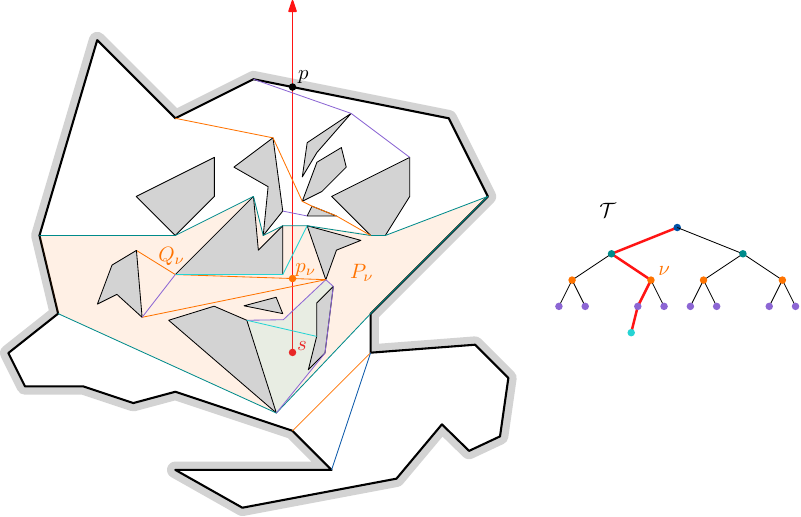}
    \caption{Schematic illustration of the ray shooting data
      structure. The point $p_\nu$ on shortest path $Q_\nu$ is among
      the first three edges of level $\geq \ell_\nu$ intersected by
      the ray.
    }
    \label{fig:colored_ray_shooting}
  \end{figure}

  Consider the nodes on the root-to-leaf path corresponding to a query
  with a point $s$, and let $\overline{sp}$ be the maximal vertically
  upward segment in $P$. Furthermore, let $\nu$ be a node on this path,
  and let $Q_\nu \in \Q_\nu$ be a shortest path that contributes an
  intersection point $p_\nu$ to $Y(s)$. It follows that
  $p_\nu = \overline{sp} \cap e_\nu$, for some for some edge
  $e_\nu \in Q_\nu$ that (thus) has level at least $\ell_\nu$; indeed
  $Q_\nu$ lies inside the subpolygon $P_\nu$ of $P$, and thus $p_\nu$
  must lie on $\overline{sp}$. Observe that, since the shortest paths of
  a node $\mu$ actually split $P_\mu$ into the subpolygons of its
  children, the interior of the subsegment $\overline{sp_\nu}$
  intersects no edges of color more than $\ell_\nu$, and the only
  intersection points with edges of color $\ell_\nu$ are with the (at
  most two) other shortest paths in $\Q_\nu$. Since $\overline{sp_\nu}$
  is a shortest path in $P_\nu$, each such path can intersect
  $\overline{sp_\nu}$ at most once. Hence, it follows that $e_\nu$ is
  among the first three intersection points of $\overline{sp}$ (as seen
  form $s$) whose color is at least $\ell_\nu$. We thus build a data
  structure that given a vertical query segment $\overline{sp}$ allows
  us to report, for each color, the edges of that color intersected by
  $\overline{sp}$ in order.

  We can solve this problem using a partially
  persistent~\cite{driscoll89makin_data_struc_persis} data structure for
  dynamic colored range reporting.  Gupta, Janardan, and
  Smid~\cite[Theorem 3.1]{gupta95furth_resul_gener_inter_searc_probl}
  show that we can store a dynamic set of $n$ colored points in $\R^1$
  so that given a query interval $I$ one can report the $i$ distinct
  colors of the points in $I$ in $O(\log n + i)$ time, and support
  insertions and deletions in $O(\log n)$ time. Their data structure
  consists of two parts: (i) for each color $c$ a binary search tree
  $T_c$ of the points of that color, and (ii) a data structure $L$ that
  allows them to report the leftmost point of each color in $I$. By
  maintaining cross pointers from the points in $L$ into the appropriate
  points in $T_c$, we can then answer our queries in $O(\log n)$ time:
  we have at most $O(\log n)$ colors, so $i = O(\log n)$, and for every
  color we report at most $O(1)$ additional points. Both the trees $T_c$
  and the data structure $L$ (which essentially also a binary search
  tree) can be made partially
  persistent~\cite{driscoll89makin_data_struc_persis} so that queries
  still take $O(\log n + i)$ time, and updates take $O(\log n)$ time and
  space.

  We now essentially use the same approach as Sarnak and Tarjan to build
  a data structure for vertical ray
  shooting~\cite{sarnak86planar_point_locat_using_persis_searc_trees}:
  we sweep the plane using a vertical line, while maintaining the line
  segments (edges used by the shortest paths) that intersect the sweep
  line in a status structure, ordered from bottom to top. However,
  rather than using a partially persistent red black tree, we use the
  partially persistent color reporting structure described above. Since
  there are a total of $O(n)$ edges, this sweep processes $O(n)$ events,
  each of which uses $O(\log n)$ time and space. Thus, we obtain an
  $O(n\log n)$ space data structure as claimed.

  To answer a query with a vertical line segment $\overline{sp}$, we
  spend $O(\log n)$ time to find the right version of the data
  structure; i.e. the one we had at ``time'' $s_x$, and then query the
  color reporting structure with the interval $[s_y,p_y]$. As argued,
  this reports $O(\log n)$ edges; at most a constant number per color
  $\ell_\nu$. These are the points contributed by the shortest paths in
  $\Q_\nu$.
\end{proof}

Finally, we plug in that $k = \lceil \frac{1}{\eps} \rceil$, which
gives us the following result.

\improvedoracle*

\begin{remark}
  Thorup's final approach~\cite{Thorup2007} uses one additional idea:
  that one can guarantee that the boundary of a subgraph $G_\nu$ (used
  during the construction in Lemma~\ref{lem:separator_subdivision})
  intersects only $O(1)$ shortest paths from ancestor separator
  paths. Therefore, he can also store anchor points on these paths,
  and query only the shortest paths of the lowest common ancestor of
  the leaves corresponding to $s$ and $t$.

  It is not clear to us if this idea is still applicable with our
  approach. For a shortest path $B$ on the boundary of some subpolyon
  $P_\nu$, it is not clear in which (sub)polygon to place the
  path $B$ so that we can invoke
  Lemma~\ref{lem:graph_and_shortest_path_datastructure}. 
\end{remark}

\frank{say s.t. about reporting the path itself}

\section{Dynamic nearest neighbor searching}
\label{sec:Dynamic_nearest_neighbor_searching}

In this section, we develop our data structure for approximate
nearest neighbor searching with a dynamic set of sites $S$ in a
polygonal domain $P$. In
Section~\ref{sub:Computing_an_close_neighbor_in_a_triangle} we first
develop a simple solution for when $S$ is simply a set of sites in the
plane. In Section~\ref{sub:Weighted_Voronoi_Diagram} we consider how
to solve a restricted case in which we are given a shortest path $Q$
in $P$, and we have to only answer queries $q$ for which the shortest
path between $q$ and any site in $S$ intersects $Q$. Then finally, in
Section~\ref{sub:The_final_data_structure} we combine these
results with the balanced hierarchical subdivision to obtain our data
structure for the general problem.

\subsection{Dynamic Euclidean $\eps$-close neighbor searching}
\label{sub:Computing_an_close_neighbor_in_a_triangle}

Let $S$ be a dynamic set of $m$ points in $\R^2$. We present a simple
$O(km\log m)$ space data structure that given a query point $q$ can
report a site $s$ for which the distance $\dist(q,s)$ is at most
$(1+\eps)\dist(q,s^*)$ in $O(k\log m)$ time. Here, $s^*$ is a site in
$S$ closest to $q$. Our data structure supports insertions and
deletions in $O(k\log m)$ time as well. Our data structure essentially
uses the cone based approach from
Section~\ref{sub:Clarkson's_Cone_graph}; i.e. for each cone $C$ in the
cone family $\F_\eps$, we build a dynamic data structure that can report
the site $s \in S \cap C_q$ with minimum cone distance. By the
argument from Lemma~\ref{lemm:Clarkson} it follows that this then
yields an $\eps$-close site.



\begin{lemma}
  \label{lem:nearest_in_cone}
  Let $C$ be a cone, and let $S$ be a set of $m$ sites in
  $\R^2$. There is an $O(m\log m)$ space data structure that, given a
  query point $q \in \R^2$ can report the point $s \in C_q \cap S$ with minimal
  cone distance $d^C(q,s)$ in $O(\log^2 m)$ time. Insertions and
  deletions of sites are supported in amortized $O(\log^2 m)$ time.
\end{lemma}
\begin{proof}
  Assume without loss of generality that the cone direction of $C$ is
  vertical and points upwards. We project each point $s$ onto lines
  perpendicular to the cone boundaries $c_\ell$ and $c_r$ (oriented
  upward). Let $s_\ell$ and $s_r$ denote the resulting values, and
  define $s'=(s_\ell,s_r)$ and $S' = \{s' \mid s \in S\}$. Given a
  query point $q$ our goal is thus to report a point $s$ with minimal
  $y$-coordinate for which
  $s' \in [q_\ell,\infty) \times [q_r,\infty)$. See
  Figure~\ref{fig:approx_euclidean_nn} for an illustration.

  \begin{figure}[tb]
    \centering
    \includegraphics{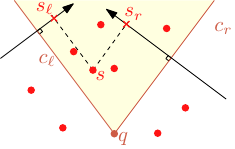}
    \caption{Our goal is to return the site $s$ with minimum
      $y$-coordinate among the sites whose projection lies in the
      range $[q_\ell,\infty) \times [q_r,\infty)$.
    }
    \label{fig:approx_euclidean_nn}
  \end{figure}

  We can thus store the projected set $S'$ in a 2D range tree~\cite{GA_book},
  in which we annotate each node $\nu$ in every associated tree with
  the (point with) minimum $y$-coordinate among all points in the
  subtree rooted at $\nu$. The total space remains $O(m\log m)$. Given
  a query, we can then find $O(\log^2 m)$ nodes in the associated
  trees that together represent exactly the subset of points from $S'$
  in the query range $[q_\ell,\infty) \times [q_r,\infty)$. Each such
  node stores the point with minimum $y$-coordinate from its subtree,
  so it suffices to just return the minimum of these values. Hence,
  queries take $O(\log^2 m)$ time.

  We implement each associated tree using a red black tree, and the
  primary tree using Anderssons' \cite{Andersson} binary trees. Since
  we can recompute the minimum value stored at some node $\nu$ from
  the minima of its children, updates to a secondary tree just take
  $O(\log m)$ time. Updates in the primary tree are handled by partial
  rebuilding~\cite{Andersson}, and thus lead to $O(\log^2 m)$
    amortized insertions and deletions.
\end{proof}

\frank{I wonder if we can do this with $O(km)$ space instead.}
\begin{corollary}\label{cor:neighbourhood-efficient}
  Let $S$ be a set of $m$ sites in $\R^2$. There is an $O(km\log m)$
  space data structure that, given a query point $q \in \R^2$ can
  report an $1/k$-close site from $S$ in $O(k\log^2 m)$ time. Insertions
  and deletions of sites are supported in amortized $O(k\log^2 m)$
  time.
\end{corollary}

\subsection{Maintaining an additively weighted Voronoi diagram}
\label{sub:Weighted_Voronoi_Diagram}

Let $S$ be the set of $m$ sites in the polygonal domain $P$, and let
$Q$ be a shortest path. We describe a dynamic data structure that,
given a query point $q \in P$, can report an $\eps$-close site
$s \in S$, provided that the shortest path from $q$ to the closest
site $s^* \in S$ intersects $Q$. More specifically, our goal is to
compute a site $s$ which achieves the following minimum

\begin{align}\label{eq:weighted-voronoi}
  \min_{s\in S}\anchorDistancePoints{q}{s} &= \min_{s\in S}\min_{a \in A_Q(q), b \in A_Q(s)} \hat{\dist}_a +
    \hat{\dist}_b + \dist(a,b).
\end{align}

If the shortest path between $s^*$ and $q$ intersects $Q$, then by
Lemma~\ref{lem:graph_and_shortest_path_datastructure} this value is indeed at most
$(1+\eps)\dist(s^*,q)$, and thus the site $s$ that achieves this minimum
is $\eps$-close to $q$. We rewrite Equation~\ref{eq:weighted-voronoi}
to $\min_{a\in A_Q(q)} L_Q(a)$, where we define
\begin{align*}
  L_Q(a) := \hat{\dist}_a + \min_{s\in S, b\in A_Q(s)}\hat{\dist}_b + \dist(a,b).
\end{align*}
Finding the optimal $b$ for the evaluation of $L_Q(a)$ is essentially
the closest site problem with additive weights
$\hat{\dist}_b$. Furthermore, we only have to consider a 1-dimensional
closest site problem as all anchor points lie on a shortest path $Q = \pi(u,v)$. This means we can map the coordinates of each point $a\in Q$ to the one-dimensional coordinate $a_x = \dist(u,a)$, this preserves distances as for all $a,b\in \pi(u,v)$ we have $\dist(a,b) = |\dist(u,b) - \dist(u,a)| = \dist(a_x, b_x)$.

We will argue that we can efficiently evaluate $L_Q(x)$ for any
arbitrary query value $x$.

\begin{figure}
  \includegraphics[page=2]{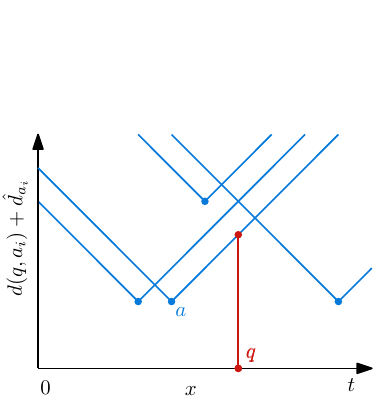}
  \quad
  \includegraphics[page=3]{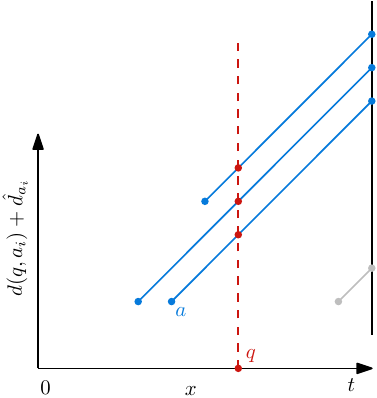}
  \caption{(a) For a query $q$ we want to find the lowest intersection
    of the line at $q = x$ with the graph of the function $d(a,\cdot)$
    of an anchor point $a \in A_Q(s)$. (b) The right half-lines of our
    V-functions. Note that the half lines intersect the vertical line
    at $q$ in the same order as at $t$. Our data structure would find
    the site to the left of $q = x$ which makes the lowest
    intersection with $x = t$. }
    \label{fig:LQ-query}
\end{figure}

\begin{lemma}\label{lem:LQ-datastructure}
    Let $t\in \mathbb{R}_{\geq 0}$ be a parameter and let $A\subset [0,t]\times \mathbb{R}_{\geq 0}$ be a dynamic set of points, where for each $a\in A$, the value $a_0$ is a one-dimensional coordinate and $a_1$ is a positive additive weight. There exists an $O(|A|)$ space data structure that can
    \begin{itemize}
      \item Report for any $q\in [0,t]$ the additively weighted closest point in $A$, i.e. the point $b$ such that $b = \argmin_{a\in A}a_1 + \dist(a_0, q)$.
      \item Support insertions and deletions to $A$ in $O(\log |A|)$ amortized time
    \end{itemize}
\end{lemma}
\begin{proof}
  For each point $a \in A$ we consider the distance function
  $f(a,x) = a_1 + \dist(a_0, x)$. The graph of the this function
  is ``V''-shaped, with its lowest point at $(a_0, a_1)$. We store the
  right half-lines and left half-lines of all these V-shaped functions
  separately, we will show how to store the right half-lines in such a
  way that we can query a coordinate $q$ to find the lowest half-line
  at that coordinate. We store the left half-lines
  symmetrically. These data structures allow us to find the additively
  weighted closest point in $A$ by taking the minimum of the results
  of our two half-line data structures, see Figure~\ref{fig:LQ-query}.

  All right half-lines have the same angle ($\pi/4$ compared to the
  $x$-axis) and they therefore have a clear ordering based on the
  intercept of their supporting line. We find this ordering by
  computing the intersection of each of these half-lines with the
  vertical line at $x = t$. Let $a_2$ be the $y$-coordinate of this
  intersection point. We store each point $a\in A$ as $(a_0, a_2)$ in
  a \emph{priority search tree}~\cite[Ch 10.2]{GA_book}. Such a tree
  supports insertions and deletions of points in $O(\log |A|)$ time
  (e.g. by implementing the base tree using a red black tree). Given a
  query interval $X$ the priority search tree can return a point with
  minimum $y$-coordinate among $A \cap (X \times \R)$ in $O(\log |A|)$
  time.

  So, to find the lowest right half-line at a coordinate $q$ we query
  the priority search tree for the lowest point in the interval
  $[0,q]$. Clearly a right half-line passes over the coordinate $q$ if
  and only if $a_0\in [0,q]$. As each of the right half-lines are
  parallel, the lowest of these lines is exactly the half-line which
  has the lowest intersection with the line at $x = t$ and therefore
  the smallest $a_2$ value.
\end{proof}

We can now describe a data structure to efficiently evaluate
Equation~\ref{eq:weighted-voronoi}.

\begin{lemma}
  \label{lem:eps_closest_via_shortest_path}
  Let $P$ be a polygonal domain with $n$ vertices, let $Q$ be a
  shortest path in $P$, and let $S$ be a dynamic set of $m$ sites in
  $P$. There exists an $O(nk + mk)$ space data structure such that:
  \begin{itemize}
  \item For any point $q\in P$ we can find $\min_{s\in S}\anchorDistancePoints{q}{s}$
    and the site which achieves this minimum in
    $O(k\log n + k^2 + k\log (mk))$ time.
  \item We can insert sites in amortized $O(k\log n + k^2 + k\log (mk))$
    time, and delete sites in amortized $O(k\log(mk))$ time.
  \end{itemize}
  Constructing an initially empty data structure takes $O(nk^2\log n\log(kn))$ time.
\end{lemma}
\begin{proof}
  We preprocess $P$ and $Q$ using
  Lemma~\ref{lem:graph_and_shortest_path_datastructure}; i.e. in
  $O(nk^2\log n\log(kn))$ time, we build an $O(nk)$ space data
  structure on $Q$ that allows us to efficiently compute, for each
  site $s \in S$, a set $A_Q(s)$ of $O(k)$ anchor points on $Q$. We
  will maintain these anchor points in the data structure from
  Lemma~\ref{lem:LQ-datastructure}. In particular, let $u$ be one of
  the endpoints of $Q$, for each anchor point $a \in A_Q(s)$ we then
  store the tuple $(\dist(u, a), \hat{\dist}_a)$. We store a total of
  $O(mk)$ anchor points, so the Lemma~\ref{lem:LQ-datastructure} data
  structure uses $O(mk)$ space. Hence, the total space usage is
  $O(nk+mk)$.

  When we insert a new site $s \in S$, we query the
  Lemma~\ref{lem:graph_and_shortest_path_datastructure} data structure
  to find the anchor points in $O(k\log n + k^2)$ time, and then
  insert them into the Lemma~\ref{lem:LQ-datastructure} data
  structure. Hence, this takes a total of
  $O(k\log n + k^2 + k\log (km))$ time. To delete a site $s$, we
  simply remove its $O(k)$ anchor points from the
  Lemma~\ref{lem:LQ-datastructure} data structure in $O(k\log (km))$
  time.

  To answer a query $q$, we use the
  Lemma~\ref{lem:graph_and_shortest_path_datastructure} data structure
  to find $O(k)$ additional anchor points $A_Q(q)$ in $O(k^2\log n)$
  time. For each of these anchor points $a$, we evaluate $L_Q(a)$
  using the Lemma~\ref{lem:LQ-datastructure} data structure in
  $O(\log(mk))$ time, and return the overall minimum value
  $\min_{s\in S}\anchorDistancePoints{q}{s} = \min_{a\in A_Q(q)}L_Q(a)$. Hence, we
  correctly answer a query in $O(k^2\log n + k\log(mk))$ time in
  total.
\end{proof}

\subsection{The final data structure}
\label{sub:The_final_data_structure}

Our data structure for general approximate nearest neighbor queries
uses the same approach as in
Section~\ref{sec:An_improved_data_structure_for_distance_queries}. Our
data structure consists of the balanced hierarchical separator \T from
Lemma~\ref{lem:separator_subdivision}, in which we store associated
structures at each node. We again preprocess the triangulation formed
by the leaves for efficient point location, and store a pointer from
each edge $e$ to the highest shortest path $Q$ in some $\Q_\nu$ that
uses edge $e$.

For each internal node $\nu$, and for each shortest path
$Q \in \Q_\nu$, we now build an associated
Lemma~\ref{lem:eps_closest_via_shortest_path} data structure on
$Q$. This structure will store a subset $S_{\nu,Q}$ of the sites. In
particular, the sites in $S \cap P_\nu$ that did not lie on shortest
paths of ancestors of $\nu$.

For each leaf $\nu$ of \T, which represents some triangle $P_\nu$, we
store (a subset $S_\nu$ of) the sites in $S \cap P_\nu$ in the data
structure from Corollary~\ref{cor:neighbourhood-efficient}. We again
store only the sites that did not already appear on shortest paths
stored at ancestors of $\nu$.

Finally, we build the data structure from
Lemma~\ref{lem:improved_query_time} that allows us to efficiently
compute the sets $X_Q(s)$ on any root-to-leaf path.




\subparagraph{Space and construction time.} The space usage of the
Lemma~\ref{lem:improved_query_time} data structure is $O(nk\log
n)$. As before, the total size of all polygons $P_\nu$ over all nodes
$\nu$ in \T is $O(n\log n)$. Hence, their total contribution to the
Lemma~\ref{lem:eps_closest_via_shortest_path} data structures is
$O(nk\log n)$. Each site in $S$ lies in the interior of at most one
subpolygon $P_\nu$ per level, or on the shortest paths in $\Q_\nu$ for
at at most one node $\nu$. Therefore, the total space of all
Lemma~\ref{lem:eps_closest_via_shortest_path} data structures is thus
$O(nk\log n + mk\log n)$. As each site also appears in at most one
leaf, so the total space of the
Corollary~\ref{cor:neighbourhood-efficient} data structures is
$O(mk\log m)$. Hence, our data structure uses
$O(nk\log n + mk\log n + mk\log m)$ space.

Constructing \T and the point location structure takes $O(n\log n)$
time. Constructing the Lemma~\ref{lem:improved_query_time} data
structure takes $O(nk\log n)$ time. This is all dominated by the time
to construct the Lemma~\ref{lem:eps_closest_via_shortest_path} data
structures. As before, the total complexity of all subpolygons is
$O(n\log n)$, so the total preprocessing time is
$O(nk^2\log^2 n\log(kn))$.

\subparagraph{Answering a query.} Let $q \in P$ be a query point. Our
goal is to report a site $s \in S$ that is $\eps$-close to $q$.

We point locate to find a triangle $\Delta_q$ that contains $q$, and
compute the sets $X_Q(s)$ for all nodes on the path $\T(q)$ towards
this leaf using the Lemma~\ref{lem:improved_query_time} data
structure. For each shortest path $Q \in \Q_\nu$ of each node
$\nu \in \T(q)$ along this path, we then query the
Lemma~\ref{lem:eps_closest_via_shortest_path} data structure. Each
such query returns a site $s_{\nu,Q}$ that minimizes
$\anchorDistancePoints{q}{s_{\nu,Q}}$ among the sites in $S_{\nu,Q}$ (as well as the
actual distance estimate $\anchorDistancePoints{q}{s_{\nu,Q}}$). Finally, we query
the Corollary~\ref{cor:neighbourhood-efficient} data structure of the
leaf triangle $\Delta_q$, to compute an $\eps$-close site from the
subset of sites stored there. We report the site $s$ with minimum
distance estimate among our candidates.

Let $s^*$ be site in $S$ closest to $q$. We now argue that the site
$s$ and the distance estimate $\hat{\dist}(s,q)$ that we return indeed
satisfies
$\dist(q,s^*) \leq \dist(q,s) \leq \hat{\dist}(s,q) \leq
(1+\eps)\dist(q,s^*)$.

If $s^*$ also lies in $\Delta_q$, then
Corollary~\ref{cor:neighbourhood-efficient} guarantees that $s$ is
within distance $(1+\eps)\dist(q,s^*)$. Otherwise, the shortest path
$\geod(q,s^*)$ from $q$ to $s^*$ must intersect one of the shortest
paths $Q \in \Q_\nu$ on one of the nodes on the root-to-leaf path
$\T(q)$. In particular, consider a shortest path of the highest such
node $\nu$. It follows that $s^*$ is one of the sites in $S_{\nu,Q}$,
and thus Lemma~\ref{lem:eps_closest_via_shortest_path} guarantees that
the candidate $s_{\nu,Q}$ has the minimum $\anchorDistancePoints{q}{s_{\nu,Q}}$ value
among all sites in $S_{\nu,Q}$. Hence, our distance estimate satisfies
$
  \hat{\dist}(q,s) \leq \anchorDistancePoints{q}{s_{\nu,Q}} \leq
  \anchorDistancePoints{q}{s^*}.
$
Lemma~\ref{lem:graph_and_shortest_path_datastructure} tells us that
$\anchorDistancePoints{q}{s^*} \leq (1+\eps)\dist(q,s^*)$ as
desired. Finally, since our distance estimate $\hat{\dist}(q,s)$
corresponds to the length of an obstacle avoiding path, we also have
$\dist(s,q^*) \leq \dist(s,q) \leq \hat{\dist}(q,s)$, and thus $s$ is
indeed $\eps$-close to $q$.

Finding the triangle containing $q$, and traversing the root-to-leaf
path corresponding to this triangle takes $O(\log n)$ time. Querying
all $O(\log n)$ Lemma~\ref{lem:eps_closest_via_shortest_path} data structures along
this path would take a total of
$O(k\log^2 n + k^2\log n + k\log (mk)\log n)$ time. Finally, the query
to the Corollary~\ref{cor:neighbourhood-efficient} data structure
stored at the leaf takes $O(k\log^2 m)$ time. As before, the
$O(k\log^2 n)$ term in the query time is only due to computing the
sets $X_Q(s)$ in every node separately. By querying the
Lemma~\ref{lem:improved_query_time} data structure, we can again
reduce this term to $O(k\log n)$. This reduces the total time required
to query the Lemma~\ref{lem:eps_closest_via_shortest_path} data
structures to $\querycomplexityk$. As this term still
dominates the other terms, this is also the final query time.

\subparagraph{Updates.} We first describe the process for inserting a new
site $s$ into $S$. We point locate to find a triangle $P_\nu$ that
contains $s$. When $s$ lies in the interior of such a triangle, then
it must lie in the interior of all ancestor subpolygons $P_\mu$ as
well, so we insert $s$ into the
Corollary~\ref{cor:neighbourhood-efficient} data structure of the leaf
$\nu$, and the Lemma~\ref{lem:eps_closest_via_shortest_path} data
structures along the root-to-leaf path. If $s$ lies on an edge $e$ of
$P_\nu$ that is also an edge in a shortest path in one of the
separators, we use the pointer stored at this edge to jump to the
shortest path $Q \in \Q_{\nu'}$ of the highest such node $\nu'$, and
insert $s$ in the Lemma~\ref{lem:eps_closest_via_shortest_path} data
structures of $\nu'$ and its ancestors.

We have to insert $s$ into at most one
Corollary~\ref{cor:neighbourhood-efficient} data structure, (in
$O(k\log^2 m)$ time), and at most $O(\log n)$
Lemma~\ref{lem:eps_closest_via_shortest_path} data structures. Naively
inserting them into every such structure then gives a total insertion
time of $O(k\log^2 n + k^2\log n + k\log (mk)\log n)$. As with the
queries, we can compute the sets $X_Q(s)$ in $O(k\log n)$ time in
total using Lemma~\ref{lem:improved_query_time}. This then gives us a
total insertion time of $\updatecomplexityk$.

When we delete a site, we simply delete the site from all places where
it was inserted. This takes a total of $\deletioncomplexityk$
time. Plugging in $k=\lceil 1/\eps \rceil$ gives us our main
result.

\mainresult*

\bibliography{bibtex}

\appendix

\end{document}